\documentclass[12pt]{article}
\usepackage[margin=1in]{geometry}
\usepackage{fullpage}
\usepackage{tgtermes}
\usepackage[T1]{fontenc}
\usepackage[colorlinks,citecolor=blue,linkcolor=blue,urlcolor=red,pagebackref]{hyperref}
\usepackage{graphicx}
\usepackage{mathtools}
\usepackage{amsfonts,amsmath,amsthm,amssymb,dsfont}
\usepackage{xfrac,nicefrac}
\usepackage{mathdots}
\usepackage{bm,dsfont}
\usepackage{url}
\usepackage{paralist}
\usepackage{enumerate}
\usepackage[normalem]{ulem}
\usepackage{xspace}
\xspaceaddexceptions{]\}}
\usepackage[capitalise]{cleveref}
\usepackage{comment}
\usepackage{tabu}
\usepackage{framed}
\usepackage{float,wrapfig}
\usepackage{subfigure}
\usepackage{tikz,pgfplots}
\pgfplotsset{compat=1.18}
\usepackage{nicefrac}
\usepackage{thm-restate}
\usepackage[usenames,dvipsnames]{pstricks}
\usepackage[algo2e,linesnumbered,boxed,ruled,vlined]{algorithm2e}

\theoremstyle{plain}
\newtheorem{prob}{Problem}

\newtheorem{theorem}{Theorem}[section]
\newtheorem{lemma}[theorem]{Lemma}
\newtheorem{prop}[theorem]{Proposition}
\newtheorem{corollary}[theorem]{Corollary}
\newtheorem{obs}[theorem]{Observation}

\theoremstyle{definition}
\newtheorem{definition}[theorem]{Definition}

\def\ShowAuthNotes{1}
\ifnum\ShowAuthNotes=1
\newcommand{\authnote}[2]{\ \\ \textcolor{red}{\parbox{0.9\linewidth}{[{\footnotesize {\bf #1:} { {#2}}}]}}\newline}
\else
\newcommand{\authnote}[2]{}
\fi

\renewcommand{\epsilon}{\varepsilon}
\newcommand{\eps}{\varepsilon}

\newcommand{\poly}{\operatorname{\mathrm{poly}}}

\newcommand{\R}{\mathbb{R}}
\newcommand{\N}{\mathbb{N}}

\newcommand{\opt}{\mathrm{OPT}}
\newcommand{\sol}{\mathrm{SOL}}

\newcommand{\parti}{Partition\xspace}
\newcommand{\caS}{\mathcal{S}\xspace}

\newcommand{\union}{\cup}
\newcommand{\intersect}{\cap}
\newcommand{\floor}[1]{\left\lfloor #1 \right\rfloor}

\DeclarePairedDelimiter\abs{\lvert}{\rvert}%
\DeclarePairedDelimiter\norm{\lVert}{\rVert}%

\makeatletter
\let\oldabs\abs
\def\abs{\@ifstar{\oldabs}{\oldabs*}}
\let\oldnorm\norm
\def\norm{\@ifstar{\oldnorm}{\oldnorm*}}
\makeatother
 
 \usepackage{accents}
\renewcommand\bar[1]{\accentset{\rule{.4em}{.6pt}}{#1}}

\title{$(1 - \eps)$-Approximation of Knapsack in Nearly Quadratic Time}

\author{  Xiao Mao \\  Stanford University \\  \texttt{xiaomao@stanford.edu} \\}

\date{}

\begin{document}
	\maketitle
	\begin{abstract}
	 \emph{Knapsack} is one of the most fundamental problems in theoretical computer science. In the $(1 - \eps)$-approximation setting, although there is a fine-grained lower bound of $(n + 1 / \eps) ^ {2 - o(1)}$ based on the $(\min, +)$-convolution hypothesis ([K{\"u}nnemann, Paturi and Stefan Schneider, ICALP 2017] and [Cygan, Mucha, Wegrzycki and Wlodarczyk, 2017]), the best algorithm is randomized and runs in $\tilde O\left(n + (\frac{1}{\eps})^{11/5}/2^{\Omega(\sqrt{\log(1/\eps)})}\right)$\footnote{Throughout this paper, we use $\tilde O(f)$ to denote $O(f\cdot \text{poly} \log (n / \eps))$. } time [Deng, Jin and Mao, SODA 2023], and it remains an important open problem whether an algorithm with a running time that matches the lower bound (up to a sub-polynomial factor) exists. We answer the question positively by showing a deterministic $(1 - \eps)$-approximation scheme for knapsack that runs in $\tilde O(n + (1 / \eps) ^ {2})$ time. We first extend a known lemma in a recursive way to reduce the problem to $n \eps$-additive approximation for $n$ items with profits in $[1, 2)$. Then we give a simple efficient geometry-based algorithm for the reduced problem.
	\end{abstract}

	\section{Introduction}

    \emph{Knapsack} is one of the most fundamental problems in
    computer science and mathematical optimization, and is actively being studied in fields such as integer programming and fine-grained complexity. 
In the Knapsack problem (sometimes also called 0-1 Knapsack), we are given a set $I$ of $n$ items where each item $i \in I$ has weight $w_i>0$ and profit $p_i>0$, as well as a knapsack capacity $W$, and we want to choose a subset $J\subseteq I$ satisfying the weight constraint $\sum_{j\in J}w_j \le W$ such that the total profit $\sum_{j\in J}p_j$ is maximized.

Knapsack is well-known to be hard: it appeared in Karp's original list of 21 NP-hard problems \cite{karp1972reducibility}.
To cope with NP-hardness, a natural direction is to study its \emph{approximation algorithms}. Given a parameter $\eps > 0$, and an input instance with optimal value $\opt$, a $(1-\eps)$-approximation algorithm is required to output a number $\sol$ such that $(1 - \eps)\opt \le \sol \le \opt$.
Fortunately, Knapsack is well-known to have \emph{fully polynomial-time approximation schemes (FPTASes)}, namely $(1-\eps)$-approximation algorithm that runs in $\poly(n,1/\eps)$ time, for any $\eps>0$.

\begin{table}
\centering
\caption{FPTASes for 0-1 knapsack}
\label{table:history}
\begin{tabular}{ l|l|l } 
\hline \hline
 $O\left( n \log n + ( \frac{1}{\eps} )^4 \log \frac{1 }{\eps}\right)$& Ibarra and Kim \cite{ibarra1975fast}& 1975 \\ 
 $O\left(n\log\frac{1}{\eps} + (\frac{1}{\eps})^4\right)$& Lawler \cite{lawler1979fast}& 1979 \\ 
 $O\left(n\log\frac{1}{\eps} + (\frac{1}{\eps})^3 \log^2 \frac{1}{\eps}\right)$& Kellerer and Pferschy \cite{kellerer2004improved}& 2004 \\ 
$O\left(n\log\frac{1}{\eps} + (\frac{1}{\eps})^{5/2} \log^3 \frac{1}{\eps}\right)$  (randomized)& Rhee \cite{rhee2015faster}& 2015\\
$O\left(n \log \frac{1}{\eps} + (\frac{1}{\eps})^{12/5}/2^{\Omega(\sqrt{\log(1/\eps)})}\right)$& Chan \cite{chan2018approximation} & 2018\\
$O\left(n \log \frac{1}{\eps} + (\frac{1}{\eps})^{9/4}/2^{\Omega(\sqrt{\log(1/\eps)})}\right)$ & Jin \cite{DBLP:conf/icalp/Jin19} & 2019 \\  
$O\left(n \log \frac{1}{\eps} + (\frac{1}{\eps})^{11/5}/2^{\Omega(\sqrt{\log(1/\eps)})}\right)$ (randomized) & Deng, Jin and Mao \cite{ourprevious} & 2023 \\  
$\tilde O(n + (\frac{1}{\eps})^{2})$ & \textbf{This work} & \\  
  \hline\hline
\end{tabular}
\end{table}

Reductions showed by Cygan~et~al. \cite{cygan2019problems} and  
K\"{u}nnemann et~al. \cite{kunnemann2017fine} imply that 0-1 knapsack and unbounded knapsack have no FPTAS in $O((n+\frac{1}{\eps})^{2-\delta})$ time, unless $(\min,+)$-convolution admits a truly subquadratic algorithm \cite{mucha2019subquadratic}. Experts in this field have been in pursuit of an algorithm that matches this lower bound long before such a lower bound is known. The history of FPTASes for the Knapsack problem is summarized in \cref{table:history}. 
The first algorithm with only subcubic dependence on $1/\eps$ was due to Rhee \cite{rhee2015faster}, which reduced the problem to a linear programming instance that can be solved using special randomized methods.
In 2018, Chan \cite{chan2018approximation} gave an elegant algorithm for the 0-1 knapsack problem in deterministic $O(n \log \frac{1}{\eps} + (\frac{1}{\eps})^{5/2}/2^{\Omega(\sqrt{\log(1/\eps)})})$ via simple combination of the SMAWK algorithm \cite{aggarwal1987geometric} and a standard divide-and-conquer technique. The speedup of superpolylogarithmic factor $2^{\Omega(\sqrt{\log(1/\eps)})}$ was due to progress on $(\min,+)$-convolution \cite{bremner2014necklaces,williams2014faster,chan2016deterministic}. Using an elementary number-theoretic lemma, Chan further improved the algorithm to $O(n \log \frac{1}{\eps} + (\frac{1}{\eps})^{12/5}/2^{\Omega(\sqrt{\log(1/\eps)})})$ time. Jin extended this algorithm by using a greedy lemma and extending the number theoretical lemma to multiple layers, obtaining a running time of $O(n \log \frac{1}{\eps} + (\frac{1}{\eps})^{9/4}/2^{\Omega(\sqrt{\log(1/\eps)})})$ \cite{DBLP:conf/icalp/Jin19}. Very recently, Deng, Jin and Mao designed an algorithm that used an improved greedy lemma using a result in additive combinatorics, as well as a procedure based on random partitioning, obtaining a running time of $O(n \log \frac{1}{\eps} + (\frac{1}{\eps})^{11/5}/2^{\Omega(\sqrt{\log(1/\eps)})})$, but the algorithm was randomized. In this work, we obtain a running time of $\tilde O(n + (1 / \eps) ^ 2)$ deterministically, putting an end to this decades-long effort to match the lower bound. Although our work uses somewhat complicated results from previous work as black boxes at the beginning, our core geometry-based procedure is fairly simple and is quite different from the framework of techniques made popular by Chan.  

For the related \textit{unbounded knapsack} problem, where every item has infinitely many copies, Jansen and Kraft \cite{jansen2018faster} obtained an $O(n + ( \frac{1}{\eps} )^2 \log^3 \frac{1}{\eps} )$-time algorithm; the unbounded version can be reduced to 0-1 knapsack with only a logarithmic blowup in the problem size \cite{cygan2019problems}. 

There are also some recent advances for the related problems Subset Sum and \parti. The Subset Sum problem is a special case of Knapsack, where the weight of an item is always equal to its profit. The \textit{partition} problem is an interesting special case of the Subset Sum problem where $W = \frac{1}{2}\sum_{i\in I}w_i$. For Subset Sum, the best known algorithm by Bringmann and Nakos \cite{BN21} runs in $\tilde O(n+\eps^{-2}/2^{\Omega(\sqrt{\log(1/\eps)})})$ time (improving \cite{DBLP:journals/jcss/KellererMPS03} by low-order factors). Bringmann and Nakos \cite{BN21} showed a matching lower bound based on the $(\min,+)$-convolution hypothesis. For \parti, the best FPTAS by Deng, Jin and Mao runs in $\tilde O(n+(\frac{1}{\eps})^{5/4})$ time \cite{ourprevious}. Abboud, Bringmann, Hermelin, and Shabtay \cite{abboud2019seth} showed that \parti cannot be approximated in $\poly(n)/\eps^{1-\delta}$ time for any $\delta>0$, under the Strong Exponential Time Hypothesis. 

We can see that before this work, the complexity of Subset Sum is already settled, but for Knapsack and \parti there still remain gaps between the best-known algorithm and its conditional lower bound. As we have mentioned, in this paper, we make a breakthrough that closes this gap for the Knapsack problem: 
        \begin{theorem} \label{theo:main}
 There is a deterministic algorithm for  $(1 - \eps)$-approximating Knapsack with running time $\tilde O(n  + \epsilon ^ {-2})$.
%\xnote{the log factor on $n$ is not final.}
 %\cnote{make it simpler}
        \end{theorem}

\section{Technical Overview}

\subsection{Problem Statement}
We study approximation algorithms for the 0--1 Knapsack problem. The input is a list of $n$ items $(p_1,w_1),\dots,(p_n,w_n)\in\mathbb{N}\times\mathbb{N}$ and a capacity $W\in\mathbb{N}$. The optimal value is
\[
\opt := \max_{J\subseteq[n]} \Bigl\{ \sum_{j\in J} p_j \ \Big|\  \sum_{j\in J} w_j \le W \Bigr\}.
\]
Equivalently: given a knapsack of capacity $W$ and a trove of items with weights and profits, choose a subset that maximizes total profit without exceeding the capacity.

We aim for an approximation scheme. Given an instance and a parameter $\eps\in(0,1)$, a \emph{$(1-\eps)$-approximation algorithm} outputs a value $\sol$ with $(1-\eps)\opt \le \sol\le \opt$.

\subsection{Reduction to Few Items}
A natural heuristic is to prefer items with a high profit-to-weight ratio $p_j/w_j$. For the \emph{fractional} version, the greedy strategy by $p_j/w_j$ is optimal; for 0--1 knapsack it can be very bad. For example, with capacity $W=100$ and items
\[
(p_1,w_1)=(100,60),\quad (p_2,w_2)=(75,50),\quad (p_3,w_3)=(70,50),
\]
the ratios satisfy $100/60>75/50>70/50$. Greedy picks $(100,60)$ for profit $100$, after which neither of the remaining items fits. The best choice is $(75,50)$ and $(70,50)$ for profit $145$. Here the first item has a strong ratio but its weight is not ``flexible'' enough: we either take it and block out other good items, or skip it entirely.

\paragraph{Arming heuristics with additive combinatorics.}
In our previous work \cite{ourprevious}, we used the following observation: suppose $S$ is a set of high-ratio items with many distinct weights. By results from additive combinatorics \cite{DBLP:journals/siamcomp/GalilM91,bw21}, for any target weight $W'$ in a reasonable range one can choose $S'\subseteq S$ whose total weight is very close to $W'$\footnote{In the actual proofs we apply these tools to profits rather than weights due to some technicality.}. Thus, such an $S$ offers both high ratios and sufficient weight \emph{flexibility} to interleave with the remaining (lower-ratio) items. This suggests that high-ratio items should dominate the optimum.

Concretely, we partition the items into three disjoint parts $I=X\cup Y\cup Z$:
\begin{description}
  \item[Part I] $X$ has exactly $\Delta$ distinct weights, where $\Delta$ is “just right”: large enough to give flexibility yet small enough to solve $X$ quickly;
  \item[Part II] $Y$ contains $O(1/\eps)$ items;
  \item[Part III] $Z:=I\setminus X\setminus Y$ contains only items with ratios lower than those in $X$ and $Y$.
\end{description}
We prove that the contribution from $Z$ to the total profit is bounded by a reasonably small value $B$. Hence we can solve the instance by solving the parts and then combine the answers:
\begin{description}
  \item[Part I] $X$ has few distinct weights; we can solve it in $\tilde{O}\!\left((1/\eps)^2\right)$ time;
  \item[Part II] $Y$ has $O(1/\eps)$ items and (hopefully) can be handled reasonably efficiently;
  \item[Part III] For $Z$, we only need to consider solutions with profit at most $B$, which (hopefully) enables us to solve for this part efficiently.
\end{description}
Using prior techniques, $X$ can be handled in $\tilde{O}\!\left((1/\eps)^2\right)$ time, but neither $Y$ nor $Z$ could be solved within the same bound.

A key observation is that $Z$ reduces to the \emph{same} problem but with an upper bound $B$ on achievable profit. We therefore apply the above reduction recursively. Because the feasible solutions shrink, we can afford a larger parameter $\Delta'$ in the next level and still solve for $X$ in $\tilde{O}\!\left((1/\eps)^2\right)$ time. Larger $\Delta'$ strengthens the additive-combinatorics guarantees, which in turn tightens the next profit cap $B'$, yielding an even more constrained subproblem (\cref{lemma:greedyrecurse}). After $\tilde{O}(1)$ levels the recursion terminates with a full solution.

\subsection{Knapsack with Few Items via Computational Geometry}
What remains is to approximate $Y$: knapsack with few items. To be able to combine the answers with other parts, we must approximate \emph{not only} for the given capacity $W$, but for all $W'\le W$. That is, we approximate the \emph{profit function} $f_Y$, where $f_Y(x)$ is the maximum profit under capacity $x$ (\cref{prob:knapmain}). Due to the nature of our reduction, it suffices to compute a $\delta$-additive approximation for the profit function for $\delta := \text{(sum of profits)} \times \eps$.

All prior algorithms in this line use divide-and-conquer \cite{chan2018approximation,DBLP:conf/icalp/Jin19,ourprevious}. Given items $I$, split into $L$ and $R$, approximate $f_L$ and $f_R$, and combine via the max-plus convolution $f_I=f_L\oplus f_R$. Since \cite{chan2018approximation}, we have been rounding profits to multiples of $\delta$, resulting in step functions with $O(1/\eps)$ steps, but no algorithm is known to compute the max-plus convolution of two sequences in strongly sub-quadratic time in their lengths. Thus the naive divide and conquer still costs $\tilde{O}\!\left((1/\eps)^3\right)$ time. Previous speedups use number-theoretic ideas \cite{chan2018approximation,DBLP:conf/icalp/Jin19,ourprevious} or random partitioning \cite{ourprevious}.

We instead use a \emph{2D analog} of the rounding technique from Bringmann and Nakos’s FPTAS (termed ``sparsification'') for \textsc{Subset Sum} \cite{BN21}. A $\delta$-additive approximation of a set $S$ is a set $S'$ such that: for every $x\in S$ there exists $x'\in S'$ with $x-\delta\le x'\le x$, and for every $x'\in S'$ there exists $x\in S$ with $x-\delta\le x'\le x$. To round $S\subseteq[0,B]$, they exhaustively remove any middle element $v$ from triples $u<v<w$ with $w-u\le\delta$. It is easy to see that the resulting set $S'$ has size $O(B/\delta)$. If $S_X$ and $S_Y$ are possible subset sums of disjoint sets $X$ and $Y$, that round to $\hat S_X,\hat S_Y$ respectively, then $\{x+y:x\in\hat S_X,\,y\in\hat S_Y\}$ is still a \emph{$\delta$-approximation} (not $2\delta$) of $\{x+y:x\in S_X,\,y\in S_Y\}$ \cite{BN21}. In other words, the rounding error remains $\delta$ and does not accumulate to $2\delta$.

A naive analog on profit \emph{functions} fails: if we round a step function $f$ by removing a breakpoint $v$ whenever consecutive breakpoints $u<v<w$ satisfy $f(u)<f(v)<f(w)\le f(u)+\delta$\footnote{We slightly abuse the notation and use ``breakpoint'' to refer to both the x-coordindate $u$ and the actual point $(u, f(u))$.}, the rounding error \emph{accumulates} under convolution. Consider
\[
f(x) =
\begin{cases}
-\infty, & x<0,\\
0,       & 0\le x<0.8,\\
\delta-1,& 0.8\le x<2,\\
\delta,  & x\ge 2.
\end{cases}
\]
The rounded $f'$ removes the breakpoint at $x=0.8$ (see \cref{fig:f-vs-fp}):
\[
f'(x)=
\begin{cases}
-\infty, & x<0,\\
0,       & 0\le x<2,\\
\delta,  & x\ge 2.
\end{cases}
\]
Then $(f\oplus f)(1.7)=2\delta-2$, but $(f'\oplus f')(1.7)=0$ (see \cref{fig:conv}).% Figure 1: f vs f' (no markers)
\begin{figure}[H]
\centering
\begin{tikzpicture}
  \pgfplotsset{compat=1.18}
  \pgfmathsetmacro{\del}{6} % assume \del > 2
  \begin{axis}[
    width=0.75\linewidth, height=6cm,
    xmin=0, xmax=3.1, ymin=-0.2, ymax=\del+0.6,
    axis lines=left,
    xtick={0,0.8,2,3}, ytick={0,\del-1,\del},
    yticklabels={$0$, $\delta-1$, $\delta$},
    legend style={draw=none, fill=none, at={(0.5,1.03)}, anchor=south},
    clip=false
  ]
    % f (solid)
    \addplot+[const plot, very thick, no marks] coordinates
      {(0,0) (0.8,0) (0.8,\del-1) (2,\del-1) (2,\del) (3.1,\del)};
    \addlegendentry{$f$}
    % f' (dashed)
    \addplot+[const plot, thick, dashed, no marks] coordinates
      {(0,0) (2,0) (2,\del) (3.1,\del)};
    \addlegendentry{$f'$}
  \end{axis}
\end{tikzpicture}
\caption{$f$ and its rounded $f'$.}
\label{fig:f-vs-fp}
\end{figure}
% Figure 2: (f ⊕ f) vs (f' ⊕ f') (no markers)
\begin{figure}[H]
\centering
\begin{tikzpicture}
  \pgfplotsset{compat=1.18}
  \pgfmathsetmacro{\del}{6} % assume \del > 2
  \begin{axis}[
    width=0.85\linewidth, height=6.4cm,
    xmin=0, xmax=4.6, ymin=0, ymax=2*\del+0.8,
    axis lines=left,
    xtick={0,0.8,1,1.6,2,2.8,3,4}, 
    ytick={0,\del-1,\del,2*\del-2,2*\del-1,2*\del},
    yticklabels={$0$, $\delta-1$, $\delta$, $2\delta-2$, $2\delta-1$, $2\delta$},
    legend style={draw=none, fill=none, at={(0.5,1.03)}, anchor=south},
    clip=false
  ]
    % g = f ⊕ f (solid) -- breakpoint at 0.8
    \addplot+[const plot, very thick, no marks] coordinates
      {(0,0)
       (0.8,0) (0.8,\del-1)
       (1.6,\del-1) (1.6,2*\del-2)
       (2.8,2*\del-2) (2.8,2*\del-1)
       (4,2*\del-1) (4,2*\del)
       (4.6,2*\del)};
    \addlegendentry{$f\oplus f$}

    % g' = f' ⊕ f' (dashed)
    \addplot+[const plot, thick, dashed, no marks] coordinates
      {(0,0) (2,0) (2,\del) (4,\del) (4,2*\del) (4.6,2*\del)};
    \addlegendentry{$f'\oplus f'$}
  \end{axis}
\end{tikzpicture}
\caption{Max-plus convolution: $f\oplus f$ vs.\ $f'\oplus f'$.}
\label{fig:conv}
\end{figure}

\paragraph{Warm-up: an idealized rounding scheme.} Our core (idealized) observation is that we should remove the breakpoint at $v$ \emph{only} when $w-u\le\delta$ \emph{and} the point $(v,f(v))$ lies \emph{below} the line segment between $(u,f(u))$ and $(w,f(w))$. As a warm-up, pretend for the moment that this idealized rounding is already stable under max-plus convolution. Under this rounding, any triple of consecutive breakpoints $x<y<z$ for which the polyline $(x,f'(x))\!\to\!(y,f'(y))\!\to\!(z,f'(z))$ is nonconvex must satisfy $z-x\ge\delta$. Consequently, the contour of $f'$ decomposes into upper convex-hull blocks that are $\delta$ apart, so there are $O(\mathrm{range}(f')/\delta)$ such blocks. To convolve two rounded profit functions, we convolve blockwise; each pair of \emph{upper} convex blocks reduces to their Minkowski sum and is computable efficiently.

In the actual algorithm, the rounding is slightly more involved. We search for two $x$-breakpoints $u<w$ such that (see \cref{fig:apex-block} for an illustration):
\begin{itemize}
  \item $f(w)-f(u)\le\delta$;
  \item $(u,f(u))$ and $(w,f(w))$ are \emph{apexes} (each lies strictly above the segment connecting its two neighboring breakpoints);
  \item every breakpoint $v'$ with $u<v'<w$ lies below the segment connecting $(u,f(u))$ and $(w,f(w))$.
\end{itemize}

\begin{figure}[H]
\centering
\begin{tikzpicture}
  \pgfplotsset{compat=1.18}
  \begin{axis}[
    width=0.75\linewidth, height=4.6cm,
    xmin=0, xmax=5.1, ymin=0, ymax=1.2,
    axis lines=left,
    xtick=\empty, ytick=\empty,
    xlabel={}, ylabel={},
    clip=false
  ]
    % f: step function (no marks on the step itself)
    \addplot+[const plot, very thick, no marks]
      coordinates {(0,0.10) (1,0.40) (2,0.45) (3,0.60) (4,0.90) (5,1.10)};

    % chord between u=1 and w=4 (collapse block)
    \addplot[thick, dashed, domain=1:4, samples=2]
      ({x},{0.40 + (0.90-0.40)/3*(x-1)});

    % neighbor chords to demonstrate apexes
    \addplot[thick, dotted, domain=0:2, samples=2] ({x},{0.10 + 0.175*x});          % neighbors of u
    \addplot[thick, dotted, domain=3:5, samples=2] ({x},{0.60 + 0.25*(x-3)});        % neighbors of w

    % breakpoint markers (minimal)
    \addplot[only marks, mark=*, mark size=1.7pt]
      coordinates {(1,0.40) (2,0.45) (3,0.60) (4,0.90)};

    % tiny labels only for u and w
    \node[anchor=south] at (axis cs:1,0.40) {\scriptsize $u$};
    \node[anchor=south] at (axis cs:4,0.90) {\scriptsize $w$};
  \end{axis}
\end{tikzpicture}
\caption{Collapsible block $[u,w]$. Dotted: neighbor chords showing $u$ and $w$ are apexes; dashed: chord over $[u,w]$; interior breakpoints lie below it.}\label{fig:apex-block}
\end{figure}

We show that under these conditions, we may remove every breakpoint $v'$ with $u<v'<w$ (see \cref{def:fgen}) while actually avoiding the accumulation of rounding errors under convolution (\cref{lemma:geomain}). This yields $O(\mathrm{range}(f')/\delta)$ blocks that \emph{alternate} between upper- and lower-convex hulls (\cref{lemma:reduce}). Using monotonicity, we also convolve efficiently when a lower-convex block is involved (\cref{lemma:combine}). Altogether, the divide-and-conquer runs in $\tilde{O}\!\left((1/\eps)^2\right)$ time (\cref{lemma:knapmain}).

    \section{Preliminaries}
    %\cnote{need to polish polish}
    \label{sec:prelim}

We write $\N = \{0,1,2,\dots\}$ and $\N^+ = \{1,2,\dots\}$. For $n\in \N$ we write $[n] = \{1,2,\dots,n\}$. For a finite set $S$ we let $\abs{S}$ be the size of $S$ (i.e. the number of elements in $S$). For a sequence $A = \{a_i\}$, we let $\abs{A}$ be the length of $A$ and for an interval $[l, r] \in [1, \abs{A}]$ we write $A_{l, r}$ as a shorthand for the subsequence $\{a_l, a_{l + 1}, \cdots, a_r\}$.

\subsection{Problem Statements}

    In the Knapsack problem, the input is a list of $n$ items $(p_1,w_1),\dots,(p_n,w_n) \in \N\times \N$ together with a knapsack capacity $W \in \N$, and the optimal value is 
    \[\opt:= \max_{J\subseteq [n]} \Big \{\sum_{j\in J }p_j \, \Big \vert \,  \sum_{j\in J}w_j \le W \Big \}.\]

Given a Knapsack instance and a parameter $\eps \in (0,1)$, an \emph{$(1-\eps)$-approximation algorithm} is required to output a number $\sol$ such that $(1 - \eps)\opt \le \sol \le \opt$.

In this problem, we can assume $n =O(\eps^{-4})$ and hence $\log n = O(\log \eps^{-1})$.  For larger $n$,  Lawler's  algorithm \cite{lawler1979fast} for Knapsack in $O(n\log \frac{1}{\eps} + (\frac{1}{\eps})^4)$ time is already near-optimal. 

We will sometimes describe algorithms with approximation ratio $1-O(\eps)$ (or $1 - \tilde O(\eps)$), which can be made $1-\eps$ by scaling down $\eps$ by a constant factor (or a logarithmic factor) at the beginning.

\subsection{Knapsack Problem and Profit functions}
\label{sec:prelimknap}

In the knapsack problem,  assume $0<w_i\le W$ and $p_i>0$ for every item $i$. Then a trivial lower bound of the maximum total profit is $\max_j p_j$. At the beginning, we can discard all items $i$ with $p_i \le \frac{\eps}{n} \max_j p_j$, reducing the total profit by at most $\eps \max_j p_j$, which is only an $O(\eps)$ fraction of the optimal total profit. Therefore, we can assume $\frac{\max_j p_j}{\min_j p_j} \le \frac{n}{\eps}$.

For a set $I$ of items, 
we use $f_I$ to denote the \emph{profit function} defined as:
\begin{equation*}
f_I(x) = \max \Bigg \{ \sum_{i\in J} p_i  : \sum_{i \in J} w_i\le x, \;\; J \subseteq I \Bigg\}
\end{equation*}
over $x\in [0,+\infty)$. 

%\cnote{do u like fancier font $\mathfrak{f}_I$.}
We adopt the terminology of Chan \cite{chan2018approximation} with some modification. We slightly abuse the term profit function to refer to any monotone non-decreasing step function over some domain range $x \in [l, r]$. The \textit{complexity} of a monotone step function refers to the number of its steps. The $(\max, +)$-convolution of two functions $f$ with domain $x \in [l_f, r_f]$ and $g$ with domain $x \in [l_g, r_g]$ is a function with domain $x \in [l_f + l_g, r_f + r_g]$ and is defined to be: 
$$(f\oplus g)(x) = \max_{x_f \in [l_f, r_f], x_g \in [l_g, r_g], x_f + x_g = x}f(x_f)+g(x_g).$$ 

The contour of a profit function with domain $x \in [x_1, x_k]$ can be described by a set $P(f)$ of $(2k - 1)$ 2-D points $$\{(x_1, y_1), (x_2, y_1), (x_2, y_2), (x_3, y_2), \cdots, (x_k, y_{k - 1}), (x_k, y_k)\}$$ with $x_1 < x_2 < \cdots < x_k$ and $y_1 < y_2 < \cdots < y_k$, where $f(x) = y_i$ for $x_i \le x < x_{i + 1} (1 \le i \le n)$.

Let $I_1,I_2$ be two disjoint subsets of items, and $I = I_1 \union I_2$. It is straightforward to see that $f_{I} = f_{I_1} \oplus f_{I_2}$.

For each item $i \in I$, we let its \emph{unit profit} be $p_i / w_i$.

    \subsection{$(1-\delta,\Delta)$ approximation up to $t$}
    We sometimes uses the notion of \emph{$(1-\delta,\Delta)$-approximation up to $t$}, as stated below.

    \begin{definition}[Approximation for Profit Functions]
          For functions $\tilde f, f$ with identical domains and real numbers $t,\Delta \in \R_{\ge 0}, \delta \in [0,1)$, we say that  $\tilde f$ is a  \emph{$(1-\delta,\Delta)$ approximation of $f$ up to $t$}, if for all $x$ in their domain,
          \[ \tilde f(x)\le f(x)\]
          holds, and 
          \[ \tilde f(x) \ge (1-\delta)f(x) - \Delta\]
          holds whenever $f(x)\le t$.
    \end{definition} 

For the case of $t=+\infty$, we simply omit the phrase ``up to $t$.''

We also refer to $(1,\Delta)$ approximation as \emph{$\Delta$-additive approximation}, and refer to $(1-\delta,0)$ approximation as \emph{$(1-\delta)$-multiplicative approximation}, or simply \emph{$(1-\delta)$ approximation}.

%We have the following simple facts regarding approximating merged profit functions.
%\begin{prop}
%\label{prop:merge}
%For $i\in \{1,2\}$, suppose $A_i$ is a $(1-\delta,\Delta_i)$ approximation of $\caS(X_i)$ up to $t$.
%Then, $(A_1+A_2) \cap [0,t]$ is a $(1-\delta,\Delta_1+\Delta_2)$ approximation of $\caS(X_1\uplus X_2)$ up to $t$.
%\end{prop}
%\begin{proof}
%For any $b\in \caS(X_1\uplus X_2) \cap [0,t]$ where $b=b_1+b_2$ for $b_i\in \caS(X_i) \cap [0,t]$ ($i\in \{1,2\}$), there exits $a_i \in A_i$ such that $(1-\delta)b_i - \Delta_i \le a_i %\le b_i$. Hence, $a_1+a_2\le b_1+b_2=b \le t$, and $a_1+a_2 \ge (1-\delta)b_1 - \Delta_1+(1-\delta)b_2 - \Delta_2  = (1-\delta)b - (\Delta_1+\Delta_2)$.

%The converse direction can be verified similarly.
%\end{proof}

%The following fact can be proved similarly.

We now prove the following fact:
\begin{prop}
\label{prop:merge2}
Let $I_1,I_2$ be two disjoint subsets of items.
For $i\in \{1,2\}$, suppose $\tilde f_i$ is a $(1-\delta,\Delta_i)$ approximation of the profit function $f_{I_i}$ up to $t$.
Then, $f = \tilde f_1 \oplus \tilde f_2$ is a $(1-\delta,\Delta_1+\Delta_2)$ approximation of $f_{I_1\union I_2}$ up to $t$.
\end{prop}
\begin{proof}
For any $b\in [0,t]$ where $f(b)=f_1(b_1)+f_2(b_2)$ for $b_1 + b_2 = b$ and $b_1, b_2 \in [0, t]$, we have for $i \in [1, 2]$, $(1-\delta)f_i(b_i) - \Delta_i \le f_{I_i}(b_i)$. Since  $I_1$ and $I_2$ are disjoint, $f_{I_1\union I_2}(b)\ge f_{I_1}(b_1)+f_{I_2}(b_2)$. Hence, $f_{I_1}(b_1)+f_{I_2}(b_2) \le (1-\delta)f_1(b_1) - \Delta_1 + (1-\delta)f_2(b_2) - \Delta_2 = (1 - \delta)(f_1(b_1) + f_2(b_2)) - (\Delta_1 + \Delta_2) = (1 - \Delta)f(b) - (\Delta_1 + \Delta_2)$.

The converse direction can be verified similarly.
\end{proof}

Following Chan \cite{chan2018approximation} and Jin \cite{DBLP:conf/icalp/Jin19}, given a monotone step function $f$ with range contained in $\{-\infty, 0\} \union [A,B]$, one can
round the positive values of $f$ down to powers of $1/(1-\eps)$, and obtain another profit function $\tilde f$ which has complexity only $O(\eps^{-1}\log (B/A))$, and $(1-\eps)$-approximates $f$. 
In our algorithm we will always have $B/A \le \poly(n/\eps)$,  so we may always assume that the intermediate profit functions computed during our algorithm are monotone step functions with complexity $\tilde O(\eps^{-1})$ after rounding.

%\subsection{Geometry-Related Definitions}
%For three points on 2-D Cartesian plane $u = (u_x, u_y)$, $v = (v_x, v_y)$ and $w = (w_x, w_y)$, we say $w$ is \emph{to the right} of $u \rightarrow v$ if $$(w_x - u_x)(v_y - u_y) - (w_y - u_y)(v_x - u_x) \ge 0.$$ Visually, $w$ is ``to the right'' of the vector $u \rightarrow v$. For two points $a = (a_x, a_y), b = (b_x, b_y)$, define $a + b = (a_x + b_x, a_y + b_y)$ and $a - b = (a_x - b_x, a_y - b_y)$.

%A set of points in a Euclidean space is defined to be \emph{convex} if it contains the line segments connecting each pair of its points. Let $P = \{(x_i, y_i)\} \in (\R ^ 2) ^ n$ be a set of $n$ points on a 2-D Cartesian plane. The convex hull of $P$, $\ch(P)$, is the (unique) minimal convex set containing $P$. It is well-known that the convex hull of $P$ is a polygon with at most $n$ vertices, and the convex hull of $P$ can be computed in $O(n \log n)$ time (see e.g. \cite{GRAHAM1972132}).

%Given two point sets $A$ and $B$, their \emph{Minkowski sum} is defined as $$A + B = \{a + b \mid a \in A, b \in B\}$$. Given two convex polygons $P$ and $Q$ with $\abs{P}$ and $\abs{Q}$ vertices, it is well-known that $P + Q$ is a convex polygon with at most $\abs{P} + \abs{Q}$ vertices and can be computed in $O(n + m)$ time (see e.g. \cite{fogel2012cgal}).

\section{Summary of Previous Techniques}
\subsection{Known Lemmas}
By known reductions (e.g., \cite{chan2018approximation,DBLP:conf/icalp/Jin19,ourprevious}), we can focus on solving the following cleaner problem, which already captures the main difficulty of knapsack:
\begin{restatable}{prob}{proknaprestate}
\label{probknap}
Assume $\eps\in (0,1/2)$ and $1/\eps \in \N^+$. Given a list $I$ of items $(p_1,w_1),\dots, (p_n,w_n)$ with weights $w_i\in \N$ and profits $p_i$ being multiples of $\eps$ in the interval $[1,2)$, compute a profit function that  $(1-\eps)$-approximates $f_I$ up to $2/\eps$. 
\end{restatable}

\begin{restatable}{lemma}{redknap}
	\label{lem:reduction-knap} If for some $c\ge 2$, \cref{probknap} can be solved in $\tilde O(n + 1/\eps^c)$ time, then $(1-\eps)$-approximating \emph{Knapsack} can also be solved in $\tilde O(n + 1/\eps^c)$ time. 
\end{restatable}
The proof of \cref{lem:reduction-knap} can be seen in the appendix.

The following useful lemma allows us to merge multiple profit functions, which was proven by Chan using divide-and-conquer and improved algorithms for $(\min,+)$-convolution \cite{bremner2014necklaces,williams2014faster,chan2016deterministic}.
\begin{lemma}[\text{\cite[Lemma 2(i)]{chan2018approximation}}]
\label{lemma:dc}
Let $f_1,\dots, f_m$ be monotone step functions with total complexity $N$ and
ranges contained in $\{-\infty, 0\} \cup [A, B]$. Then we can compute a monotone step function that has complexity
$\tilde O(\frac{1}{\eps}\log B/A)$ and 
$(1-O(\eps))$-approximates $f_1 \oplus\dots \oplus f_m$, in $O(N) + \tilde O((\frac{1}{\eps})^2 m/2^{\Omega(\sqrt{\log(1/\eps)})} \log B/A)$ time.
\end{lemma}

The following reduction was used in \cite{ourprevious}.
\begin{lemma}
    \label{lemma:reduction}
    Given a list $I$ of $n$ items with $p_i$ being multiples of $\eps$ in interval $[1,2)$, and integer $1\le m \le n$ with $m=O(1/\eps)$, if for some $c\ge 2$ one can compute a profit function that $(m\eps)$-additively approximates $f_I$ up to $2m$ in $\tilde O(n + 1/\eps^c)$ time . Then one can $(1 - \eps)$-approximate \emph{Knapsack} in $\tilde O(n + 1/\eps^c)$ time.
\end{lemma}
\begin{proof}
    From \cref{lem:reduction-knap}, it suffices to approximate \cref{probknap}. We divide $[1, 2\eps ^ {-1})$ into $O(\log (1/\eps))$ many intervals $[m,2m)$ where $m$'s are powers of 2, and for each $m$, compute profit functions $f_m$ achieving $m\eps$-additive approximation up to $2m$. For each $x \ge 0$, let $m_x$ be such that $m_x \le f_I(x) < 2m_x$. Then $f_{m_x}(x) \in [f_I(x) - m_x\eps, f_I(x)]$. Since $f_I(x) \ge m_x$, $f_{m_x}(x) \in [f_I(x)(1 - \eps), f_I(x)]$. Thus taking the pointwise maxima of all $f_m$'s yields a $(1-\eps)$ approximation of $f_I$. Note that by using the rounding scheme mentioned before we can make the complexity of each profit function $\tilde O(1 / \eps)$. Therefore taking the pointwise maxima of all $f_m$'s takes $O(\log (1/\eps)) \times \tilde O(1 / \eps) = o(n + 1/\eps^c)$ time.
\end{proof}

As an ingredient, we use a corollary of Lemma 7 of \cite{chan2018approximation}. The corollary states that profit functions can be approximately more efficiently when there are few distinct profit values or when there is a good upper bound on profit values. We first present the lemma in its original form. A step function $f$ is $p$-uniform if the function values are of the form $-\infty, 0, p, 2p, \cdots lp$ for some non-negative integer $l$. A $p$-uniform function is pseudo-concave if the sequence of differences of consecutive $x$-breakpoints is non-decreasing.
\begin{lemma} [Lemma 7 of \cite{chan2018approximation}]
\label{lemma:chanoriginal}
Let $f_1, ..., f_m$ be monotone step functions with ranges contained in $\{-\infty, 0\} \cup [1, B]$. If every $f_i$ is $p_i$-uniform and pseudo-concave for some $p_i \in [1, 2],$ then we can compute a monotone step function that $(1 - \eps)$-approximates $\min(f_1 \oplus \cdots \oplus f_m, B)$ and complexity $\tilde{O}(\frac{1}{\eps})$ in $\tilde{O}(\sqrt{B} m / \eps)$ deterministic\footnote{The original lemma used the word "expected" since it was initially randomized due to a randomized construction for a number theoretical lemma, but Chan later derandomized it in the same paper. A simpler deterministic construction can be seen in Lemma 3.11 of \cite{ourprevious}.} time, assuming $B = \tilde O(1 / \eps).$
\end{lemma}

\begin{restatable}{corollary}{lemmachan}
\label{lemma:chan}
    Given a list $I$ of items $(p_1,w_1),\dots, (p_n,w_n)$ with weights $w_i\in \N$ and profits $p_i$ being multiples of $\eps$ in the interval $[1,2)$, if there are only $\Delta$ distinct profit values $p_i$, then for any $B = \tilde O(1 / \eps)$, one can $(1 - \eps)$-approximate $\min(f_I, B)$ (i.e. $f_I$ up to $B$) in $\tilde O(n+\Delta \sqrt{B} \eps ^ {-1})$ time.
\end{restatable}
\begin{proof}
    Note that $I = I_1 \cup I_2 \cup \cdots \cup I_{\Delta}$ where each $I_j$ contains items of the same profit value. It suffices to compute $\min(f_{I_1} \oplus \cdots \oplus f_{I_{\Delta}}, B)$. Fix $j \in [1, \Delta]$. If we let ${w ^ {\prime}}_k$ be the $k$-th smallest weight in $I_j$, then:
    \begin{itemize}
        \item For every $0 \le k \le \abs{I_j} - 1$, for $x \in \left[\sum_{1 \le k ^ {\prime} \le k}{{w ^ {\prime}}_{k ^ {\prime}}}, \sum_{1 \le k ^ {\prime} \le k + 1}{{w ^ {\prime}}_{k ^ {\prime}}}\right)$ we have $f_{I_j}(x) = kp_j$.
        \item $f_{I_j}\left(\sum_{1 \le k ^ {\prime} \le \abs{I_j}}{{w ^ {\prime}}_{k ^ {\prime}}}\right) = \abs{I_j}p_j$.
    \end{itemize} 
    It is easy to see that $f_{I_j}$ is $p_j$-uniform and pseudo-concave with $p_j \in [1, 2]$.  Furthermore, for each $j$ it suffices to compute $\min(f_{I_j}, B)$, whose range is contained in $\{-\infty, 0\} \cup [1, B]$. From \cref{lemma:chanoriginal} we can compute $\min(f_{I_1} \oplus \cdots \oplus f_{I_{\Delta}}, B)$ in $\tilde{O}(\sqrt{B} m / \eps)$ time. Constructing $\min(f_{I_j}, B)$ all $j$'s takes $\tilde O(n)$ time. The total time complexity is $\tilde O(n+\Delta \sqrt{B} / \eps)$.
\end{proof}

\subsection{Greedy Exchange Lemma}
We now introduce the Greedy Exchange Lemma previously employed in \cite{ourprevious}, used to prove the following lemma:
\begin{restatable}{lemma}{lemmaourprevious} [Lemma 3.5 of \cite{ourprevious}]
    \label{lem:greedyweaker}
        Given a list $I$ of $n$ items with $p_i$ being multiples of $\eps$ in interval $[1,2)$, and integer $1\le m \le n$ with $m=O(1/\eps)$, one can compute a profit function that $(m\eps)$-additively approximates $f_I$ up to $2m$ in $O(n+\eps^{-11/5}/2^{\Omega(\sqrt{\log (1/\eps)})} )$ time .
\end{restatable}
We will follow the framework in \cite{ourprevious}. Given items $(p_1,w_1),\dots,(p_n,w_n)$, where $p_i\in[1,2)$ are multiples of $\eps$, we sort them by non-increasing order of unit profit: $p_1/w_1\ge p_2/w_2\ge\dots \ge p_n/w_n$. 

Then, we consider prefixes of this sequence of items, and define the following measure of diversity:
\begin{definition}[$D(i)$]
   For $1\le i \le n$, let $D(i)= \min_{J} C([i]\setminus J)$, where the minimization is over all subsets $J \subseteq [i]$ with $|J|\le 2m$, and $C([i] \setminus J)$ denote the number of distinct values in $\{p_j : j\in [i] \setminus J\}$. 
   \label{def:d}
\end{definition}
We have the following immediate observations about $D(i)$:
\begin{obs}
    \begin{enumerate}
        \item For all $2\le i\le n$, $0\le D(i) - D(i-1)\le 1$.
        \item $D(i)$ (and the minimizer $J$) can be computed in $\tilde O(i)$ time by the following greedy algorithm: start with all values $p_1,p_2,\cdots,p_i$.  Repeat the following up to $2m$ times: remove the value $p_j$ with the minimum multiplicity, and add $j$ into $J$. 
    \end{enumerate}
    \label{obs:D}
\end{obs}

\begin{definition}[$i ^ {\Delta}$]
   For some parameter $\Delta$, define ${i ^ {\Delta}} \in \{1,2,\dots,n\}$ be the maximum such that $D({i ^ {\Delta}}) \le \Delta$, which can be found using \cref{obs:D} with a binary search in $\tilde O(n)$ time. 
   \label{def:idelta}
\end{definition}

In \cite{ourprevious}, the following lemma was proven. Recall that \cref{lemma:chan} guarantees a faster run time when we have a bound on profit values. The goal of the lemma is to impose a bound on the total profit items in $[n] \setminus [{i ^ {\Delta}}]$ can contribute. On a high level, it argues that, given any subset $S$ of items, we can swap out some items in $S \intersect ([n] \setminus [{i ^ {\Delta}}])$ with items in $[{i ^ {\Delta}}]$, resulting in a new subset $\tilde{S}$ that is not too much worse than $S$, but the contribution from $[n] \setminus [{i ^ {\Delta}}]$ to $\tilde S$ is small.
\begin{restatable}{lemma}{greedyexchange}[Greedy Exchange Lemma, Lemma 3.9 of \cite{ourprevious}]
\label{lemma:exchange}
     Let $S\subseteq [n]$ 
     be any item set with total profit $\sum_{s\in S}p_s \le 2m$. Let $\Delta = \omega(\eps ^ {-1 / 2} \log ^ {1 / 2}{(1 / \eps)})$ be some parameter. Let $B:=c\eps^{-1}/\Delta$ for some universal constant $c\ge 1$.
     
     Then, there exists an item set $\tilde S \subseteq [n]$, such that 
     \begin{equation}
     \label{eqn:req1}
     \tilde  p:=\sum_{s \in \tilde S \cap ([n] \setminus [{i ^ {\Delta}}])} p_s \le B,
     \end{equation}
           and
           \begin{equation}
     \label{eqn:req2}
          \sum_{s \in \tilde  S } p_s  \ge (1-\eps) \sum_{s \in   S } p_s,
           \end{equation}
          and
          \begin{equation}
     \label{eqn:req3}
          \sum_{s \in \tilde  S } w_s  \le  \sum_{s \in   S } w_s.
          \end{equation}
\end{restatable}
The proof is included in the appendix for completeness.

In \cite{ourprevious}, we have $\Delta := \lfloor\epsilon ^ {-5 / 8}\rfloor$, which satisfies the condition in \cref{lemma:exchange}.

The reduction in \cite{ourprevious} uses the following lemma for $(1-\eps)$-approximating knapsack up to a small $B$:
\begin{lemma} [Follows from Lemma 17 of \cite{DBLP:conf/icalp/Jin19}] \label{theo:jin1}
            Given a list $I$ of items $(p_1,w_1),\dots, (p_n,w_n)$ with weights $w_i\in \N$ and profits $p_i$ being multiples of $\eps$ in the interval $[1,2)$, one can $(1-\eps)$-approximate the profit function $f_I$ up to $B$ in $\tilde O(n+\epsilon ^ {-2}B ^ {1 / 3} / 2 ^ {\Omega(\sqrt{\log(1 /  \epsilon)})})$ time.        
\end{lemma}

Finally, a weaker version of \cref{lemma:knapmain} was proven in \cite{ourprevious} using random partitioning.
\begin{restatable}{lemma}{knapsackmainweakerrestate} [Lemma 3.6 of \cite{ourprevious}]
 \label{lemma:knapmainweaker}

 \Cref{prob:knapmain} can be solved in ramdomized $\tilde O(n ^ {{4}/{5}}\epsilon ^ {-{7}/{5}}/2^{\Omega(\sqrt{\log (1/\eps)})})$ time.
\end{restatable}

In \cite{ourprevious}, a randomized running time of $\tilde O(n+\eps^{-11/5}/2^{\Omega(\sqrt{\log (1/\eps)})} )$ was achieved for \cref{probknap} by combining the following lemma with \cref{lemma:reduction}. To finish our introduction for techniques in \cite{ourprevious}, we will also include its proof.
\lemmaourprevious*
\begin{proof}[Proof of \cref{lem:greedyweaker}]
    Recall that we have $\Delta := \lfloor\epsilon ^ {-5 / 8}\rfloor$, and that $i ^ {\Delta} \in \{1,2,\dots,n\}$ is the maximum such that $D(i ^ {\Delta}) \le \Delta$, which can be found using \cref{obs:D} with a binary search in $\tilde O(n)$ time.  Let $J \subset[i ^ {\Delta}]$ with $|J|\le 2m$ be the minimizer for $D(i ^ {\Delta})$.
    
    Now, we approximately compute the profit functions $f_{J}, f_{[i ^ {\Delta}]\setminus J}, f_{[n]\setminus [i ^ {\Delta}]}$ for three item sets $J, [i ^ {\Delta}]\setminus J, [n]\setminus [i ^ {\Delta}]$ using different algorithms, described as follows:
    \begin{enumerate}
        \item Use \cref{lemma:knapmainweaker} to compute $f_1$ that $(2m\eps)$-additively approximates $f_J$ in randomized\linebreak $O(m ^ {\frac{4}{5}}\epsilon ^ {-\frac{7}{5}} / 2 ^ {\Omega(\sqrt{\log(1 /  \epsilon)}}) \le O(\epsilon ^ {-\frac{11}{5}} / 2 ^ {\Omega(\sqrt{\log(1 /  \epsilon)}})$ time.
        \item By the definition of $i ^ {\Delta}$, items in $[i ^ {\Delta}]\setminus J$ have no more than $\Delta$ distinct profit values. Hence we can use \cref{lemma:chan} to compute $f_2$ that $(1-\eps)$-approximates  $f_{[i ^ {\Delta}]\setminus J}$ up to $2m$, in $\tilde O(n + \Delta \sqrt{2m}\epsilon ^ {-1}) = \tilde O(\epsilon ^ {-17 / 8})$ time.
        \item Use \cref{theo:jin1} to compute $f_3$ that $(1-\eps)$-approximates the $f_{[n]\setminus [i ^ {\Delta}]}$ up to $B=c\eps^{-1}/\Delta=\Theta(\eps^{-1}/\Delta)$ where $c$ is the universal constant in \cref{lemma:exchange}, in $\tilde O(B^{1/3}\eps^{-2}) \le \tilde O(\epsilon ^ {-17 / 8})$ time.
    \end{enumerate}
        Finally, merge the three parts $f_1,f_2,f_3$ using \cref{lemma:dc} in $\tilde O(\eps ^ {-2})$ time, and return the result.
        
    In the third part, the correctness of only computing up to $B$ is justified by \cref{lemma:exchange}, which shows that if we only consider approximating sets with total profit up to $2m$, then we can assume the items in $[n] \setminus [i ^ {\Delta}]$ only contributes profit at most $B$ \eqref{eqn:req1}, at the cost of only incurring an $(1-\eps)$ approximation factor \eqref{eqn:req2}.
    
    To analyze the error, notice that in the first part we incur an additive error of $2m\eps$. In the second and third part and the final merging step we incur $(1-O(\eps))$ multiplicative error, which turns into $O(m\eps)$ additive error since we only care about approximating up to $2m$. Hence the overall additive error is $O(m\eps)$, which can be made $m\eps$ by lowering the value of $\eps$.
\end{proof}

Our new algorithm will use a new reduction that uses the Greedy Exchange Lemma recursively. Intuitively, observe that the computation of $f_3$ is a reduced version of the original problem with a stricter upper bound, so instead of computing it directly, we can apply the procedure recursively.

\section{Main Algorithm}
    \label{sec:knapsack}

\subsection{Reduction to the Sparse Case via Recursive Greedy Exchange} \label{sec:reduction}

In this subsection, we will prove the following:
\begin{lemma}
    \label{lemma:greedy}
    Given a list $I$ of $n$ items with $p_i$ being multiples of $\eps$ in interval $[1,2)$, and integer $1\le m \le n$ with $m=O(1/\eps)$, one can compute a profit function with complexity $\tilde O(\eps ^ {-1})$ that $(m\eps)$-additively approximates $f_I$ up to $2m$ in $\tilde O(n+\eps^{-2})$ time. 
\end{lemma}
\Cref{lemma:greedy} implies \cref{theo:main} due to \cref{lemma:reduction}. Our core ingredient is \cref{lemma:knapmain} which will be proved in \cref{sec:smalln}:
\begin{restatable}{prob}{probfewitems}
    \label{prob:knapmain}
    Given a list $I$ of $n= O(1/\eps)$ items with $p_i$ being multiples of $\eps$ in interval $[1,2)$\footnote{The input profit values in the Knapsack problem are integers, but in intermediate problems, the profits can be real numbers due to rounding.}, compute a profit function that $(n\eps)$-additively approximates $f_I$.
\end{restatable}
\begin{restatable*}{lemma}{knapsackmainrestate}
    \label{lemma:knapmain}
    \Cref{prob:knapmain} can be solved in $\tilde O(\epsilon ^ {-2})$ time.
\end{restatable*}

Also recall \cref{lemma:chan}:
\lemmachan*

To prove \cref{lemma:greedy}, we will use the proof in \cref{lem:greedyweaker} in a recursive manner. We solve the following problem for a variety of parameters $B$. Assume that $C > 0$ is a universal large constant.
\begin{prob} \label{prob:greedyrecurse}
    Given a list $I$ of $n$ items with $p_i$ being multiples of $\eps$ in interval $[1,2)$, integer $1\le m \le n$ with $m=O(1/\eps)$, integer $k > 0$, as well as a number $1 \le B \le 2m$, compute a profit function that $(Ckm\eps)$-additively approximates $f_I$ up to $B$.
\end{prob}
\begin{lemma} \label{lemma:greedyrecurse}
    Let $T(n, m, B, k)$ be the running time for \cref{prob:greedyrecurse} with parameters $n, m, k$ and $B$. Then there exists $B ^ {\prime} = o(B ^ {1 / 2})$ and $n ^ {\prime} \le n$ such that:
    $$T(n, m, B) = \tilde O(n + \eps ^ {-2}) + T(n ^ {\prime}, m, B ^ {\prime}).$$
\end{lemma}
\begin{proof}
        To solve the problem for $(n, m, k, B)$, we set $$\Delta := \Theta((\eps ^ {-1} / B ^ {1 / 2}) \log{(1 / \eps)}) = \omega(\eps ^ {-1 / 2} \log ^ {1 / 2}{(1 / \eps)}).$$ Recall \cref{def:d}, \cref{obs:D} and \cref{def:idelta}. ${i ^ {\Delta}} \in \{1,2,\dots,n\}$ is the maximum such that $D({i ^ {\Delta}}) \le \Delta$, which can be found using with a binary search in $\tilde O(n)$ time.  Let $J \subset[{i ^ {\Delta}}]$ with $|J|\le 2m$ be the minimizer for $D({i ^ {\Delta}})$.

        As in the proof of \cref{lem:greedyweaker}, we approximately compute the profit functions $f_{J}, f_{[{i ^ {\Delta}}]\setminus J}, f_{[n]\setminus [{i ^ {\Delta}}]}$ for three item sets $J, [{i ^ {\Delta}}]\setminus J, [n]\setminus [{i ^ {\Delta}}]$ using different algorithms, described as follows:
        \begin{enumerate}
            \item Use \cref{lemma:knapmain} to compute $f_1$ that $(2m\eps)$-additively approximates $f_J$, in $\tilde O(\epsilon ^ {-2})$ time.
            \item By definition of ${i ^ {\Delta}}$, items in $[{i ^ {\Delta}}]\setminus J$ have no more than $\Delta$ distinct profit values. Since we only need to approximate up to $B$, we can use \cref{lemma:chan} to compute $f_2$ that $(1-\eps)$-approximates $\min(f_{[{i ^ {\Delta}}]\setminus J}, B)$ in $\tilde O(n + \Delta \sqrt{B}\epsilon ^ {-1}) = \tilde O(n + \epsilon ^ {-2})$ time.
            \item We $C(k - 1)m\eps$-additively approximate $\min(f_3, B ^ {\prime})$ for $B ^ {\prime} = c\eps^{-1}/\Delta = o(B ^ {1 / 2})$ where $c$ is the universal constant in \cref{lemma:exchange}. This is equivalent to \cref{prob:greedyrecurse} with parameters $(n ^ {\prime}, m, k - 1, B ^ {\prime})$ where $n ^ {\prime} = n - i ^ {\Delta}$, and can be done in $T(n ^ {\prime}, m, k - 1, B ^ {\prime})$ time. 
        \end{enumerate}
        
        Finally, merge the three parts $f_1,f_2,f_3$ using \cref{lemma:dc} in $\tilde O(n + \eps ^ {-2})$ time, and return the result. The total running time is $\tilde O(n + \eps ^ {-2}) + T(n ^ {\prime}, m, B ^ {\prime})$.
            
        In the third part, the correctness of only computing up to $B ^ {\prime}$ is justified by \cref{lemma:exchange}, which shows that if we only consider approximating sets with total profit up to $2m \ge B$, then we can assume the items in $[n] \setminus [{i ^ {\Delta}}]$ only contributes profit at most $B ^ {\prime}$ (Inequality \ref{eqn:req1}), at the cost of only incurring an $(1-\eps)$-approximation factor (Inequality \ref{eqn:req2}).

        To analyze the error, notice that in the first part, we incur an additive error of $2m\eps$. In the second part and the final merging step we incur $(1-O(\eps))$ multiplicative error, which turns into $O(m\eps)$ additive error since we only care about approximating up to $B \le 2m$. Hence the additive error for these parts is $O(m\eps)$. By letting the constant $C$ be large we can assume that the overall additive error is at most $Cm\eps$. Finally, the third part incurs an additive error of $C(k - 1)m\eps$. Hence the overall additive error is at most $C(k - 1)m\eps + Cm\eps = Ckm\eps$. 
 \end{proof}

 We are now ready to prove \cref{lemma:greedy}
 \begin{proof} [Proof of \cref{lemma:greedy}]
    Note that \cref{prob:greedyrecurse} is trivial when $B < 2$. Note that for any value no more than $2m$, by recursively taking its square root at most $O(\log \log m)$ times, we get a value less than $2$. Note that in \cref{lemma:greedyrecurse} we always have $B ^ {\prime} < B ^ {1 / 2}$. Thus the running time for $T(n, m, k, 2m)$ is a summation of at most $O(\log \log m)$ repetitions of the term $\tilde O(n + \eps ^ {-2})$, as long as $k$ is larger than the number of repetitions. Hence, for $k = \omega(\log \log m)$ we have:
    $$T(n, m, k, 2m) \le O(\log \log m) \times \tilde O(n + \eps ^ {-2}) = \tilde O(n + \eps ^ {-2}).$$
    Namely, we can solve \cref{prob:greedyrecurse} in $\tilde O(n + \eps ^ {-2})$ time for $(n, m, k, 2m)$ as long as $k = \omega(\log \log m)$. Thus for $k = O(\log (1 / \eps))$ , we can get an $O(m \eps \log (1 / \eps))$-additive approximation for the problem in \cref{lemma:greedy} $\tilde O(n + \eps ^ {-2})$ time, which can be made into an $m\eps$-additive approximation by lowering the value of $\eps$ by a polylogarithmic factor.
 \end{proof}

\subsection{Approximation for the Sparse Case via Geometry} \label{sec:smalln}
\subsubsection{Definitions}
In this section, we prove our core lemma:
\knapsackmainrestate*

First we recall the relevant definitions. The contour of a profit function with domain $x \in [x_1, x_k]$ with a complexity of $k$ can be described by a set\footnote{A set must contain distinct elements.} $P(f)$ of $(2k - 1)$ 2-D points $$\{(x_1, y_1), (x_2, y_1), (x_2, y_2), (x_3, y_2), \cdots, (x_k, y_{k - 1}), (x_k, y_k)\}$$ with $x_1 < x_2 < \cdots < x_k$ and $y_1 < y_2 < \cdots < y_k$, where $f(x) = y_i$ for $x_i \le x < x_{i + 1} (1 \le i \le n)$.

A set of points $P = \{(x_1, y_1), (x_2, y_2), \cdots, (x_k, y_k)\}$ satisfying $x_i \le x_{i + 1}$ and $y_i \le y_{i + 1}$ for $1 \le i \le k - 1$ and $x_{i + 2} > x_i$ for $1 \le i \le k - 2$ is called a \emph{monotone point set}. The profit function $f(P)$ with domain $[x_1, x_k]$ is defined as $(f(P))(x) = y_i$ for $x_i \le x < x_{i + 1}$ for $1 \le i < k$ and $(f(P))(x_k) = y_k$.

For a set $S$ of points in the 2D plane, let ${\min}_x{S}, \max_x{S}, {\min}_y{S}$ and ${\max}_y{S}$ be the minimum x-coordinate, the maximum x-coordinate, 
the minimum y-coordinate, and the maximum y-coordinate of points in $S$.

For a point $p$ on the 2D plane, we let $p_x$ and $p_y$ be its x-coordinate and y-coordindate respectively. For two distinct points $p$ and $q$, we say $p < q$ if either $p_x < q_x$ or $p_x = q_x$ and $p_y < q_y$. Given two points $p$ and $q$, let $r := p + q$ be the point with $r_x = p_x + q_x$ and $r_y = p_y + q_y$, and let $r := p - q$ be the point with $r_x = p_x - q_x$ and $r_y = p_y - q_Y$. Given a point $p$ and a point set $S$, let $p + S := S + p := \{p + s \mid s \in S\}$. Given two sets of points $A$ and $B$, their \emph{Minkowski sum} is defined as $A + B = \{a + b \mid a \in A, b \in B\}.$

For a monotone point set $P$, the set of segments $E(P)$ are the line segments between adjacent points in $P$: $E(P) := \{(p_i, p_{i + 1}) \mid 1 \le i < \abs{P}\}$. The region $R(P)$ is defined to be the set of points $(x, y) (x_1 \le x \le x_k, y \ge 0)$ on or below the segments $E(P)$. For $x_1 \le x \le x_k$, let $g[P](x)$ be equal to the maximum y-coordinate of a point on the upper boundary of $R(P)$ with x-coordinate equal to $x$. For any $1 \le i \le \abs{P}$, if $i \in \{1, \abs{P}\}$ or the point $p_i$ is above the line segment connecting $p_{i - 1}$ and $p_{i + 1}$, we call the point $p_i$ an \emph{apex} of $P$. If $i \notin \{1, \abs{P}\}$ and the point $p_i$ is below the line segment connecting $p_{i - 1}$ and $p_{i + 1}$, we call the point $p_i$ a \emph{base} of $P$. If $i \notin \{1, \abs{P}\}$ and the point $p_i$ is on the line segment connecting $p_{i - 1}$ and $p_{i + 1}$, we call the point $p_i$ \emph{degenerate}.

\begin{figure}[ht]
\centering 
    \begin{tikzpicture}

    % Define points
    \coordinate (A) at (1,1);
    \coordinate (B) at (2,1);
    \coordinate (C) at (2,2);
    \coordinate (C2) at (6.5,2);
    \coordinate (D) at (3,2.5);
    \coordinate (E) at (4,4);
    \coordinate (F) at (5,6);
    \coordinate (F2) at (6.5,6);
    \coordinate (G) at (6,7.33333);
    % Draw shaded area
    \fill[gray!30] (1,0) -- (A) -- (B) -- (C) -- (D) -- (E) -- (F) -- (G) -- (6,0) -- cycle;

    \node[above=2pt] at (5, 0.5) {\LARGE $R(P)$};
    
    \node[left=2pt] at (5.5, 7) {\large \red $g[P]$};
    \node[right=2pt] at (4, 3) {\large \blue $f(P)$};
    
    % Draw curve
    \draw[thick,red] (A) -- (B) -- (C) -- (D) -- (E) -- (F) -- (G);
    
    \draw[thin,blue] (A) -- (B) -- (C) -- (3, 2) -- (D) -- (4, 2.5) -- (E) -- (5, 4) -- (F) -- (6, 6) -- (G);

    \draw[dashed,black] (A) -- (C);
    
    \draw[dashed,black] (C) -- (F);
    
    \draw[dashed,black] (B) -- (D);
    
    \draw[dashed,black] (A) -- (F);
    
    \draw[dashed,black] (C) -- (C2);
    \draw[dashed,black] (F) -- (F2);
    \draw[->,black] (C2) -- (F2) node[midway, right] {$\delta$};

    % Label points
    \foreach \point/\position/\name in {A/below/A, B/below/B, C/left/C, D/above/D, E/left/E, F/above/F, G/above/G} {
        \fill[black] (\point) circle (2pt);
        \node[\position=2pt] at (\point) {\name};
    }
    
    % Draw axes
    \draw[thick,->] (0.5,0) -- (6.5,0) node[anchor=north west] {$x$};
    \draw[thick,->] (0.5,0) -- (0.5,7) node[anchor=south east] {$y$};
    
\end{tikzpicture}
\caption{Monotone Point Set $P = \{A, B, C, D, E, F, G\}$} \label{fig:mps}
\end{figure}  

For example, in \cref{fig:mps}, the set $P = \{A, B, C, D, E, F, G\}$ is a monotone point set. The region $R(P)$ is the shaded region below the thick red curve, which is the plot of $g[P]$. The plot of $f(P)$ is the thin blue line. For instance, $C$ is an apex since it is above the line connecting $B$ and $D$, and $B$ is a base since it is below the line connecting $A$ and $C$. We will see an example of a degenerate point later.

Given two monotone point sets $P_a$ and $P_b$, let $P_a \oplus P_b$ be the minimum set $V$ of points such that $R(V) = R(P_a) + R(P_b)$. It is easy to verify that $V$ is a monotone point set without degenerate points.

It is well-known \cite{compgeom:2000} that the Minkowski sum of two polygons $P$ and $Q$ with $\abs{P}$ and $\abs{Q}$ vertices is the union of $\abs{P}$ translations of $Q$ and $\abs{Q}$ translations of $P$ consisting of:
\begin{itemize}
    \item $p + Q$ for each $p \in P$, and
    \item $P + q$ for each $q \in Q$.
\end{itemize}

It is also easy to verify that $g[P_a \oplus P_b] = g[P_a] \oplus g[P_b]$.

\subsubsection{Roadmap}
We will show an analog of the ``sparsification'' technique in \cite{BN21} in 2D. We can approximate a profit function $f$ with a monotone point set. To start, we can see that $P(f)$ can approximate $f$ exactly. We now show a way to reduce the monotone point set associated with a profit function so that we can still approximate the profit function with the reduced monotone point set, and the rounding errors do not accumulate when combining two profit functions and their monotone point set approximations.
\begin{definition} [$F$-generatability] \label{def:fgen}
For a parameter $\delta := \eps n$, and a set of $m$ profit functions $F = \{f_1, f_2, \cdots, f_m\}$, we define the notion of $F$-generatability on a pair consisting of a profit function and a monotone point set based on the three operations below:
\begin{enumerate}
    \item For each $f \in F$, $(f, P(f))$ is $F$-generatable.
    \item If $(f_a, P_a)$ and $(f_b, P_b)$ are both $F$-generatable, $(f_a \oplus f_b, P_a \oplus P_b)$ is also $F$-generatable;
    \item Suppose $(f, P)$ is $F$-generatable. Suppose we perform a \emph{reduction} on $P$: find two apexes $x, y \in P$ on $P$ with $0 < y_y - x_y \le \delta$ such that all points strictly between $x$ and $y$ are below the line segment connecting $x$ and $y$, let $\bar{P}$ be $P$ with all points between $x$ and $y$ removed. Let $P ^ {\prime}$ be the minimum subset of $\bar{P}$ with $R(P ^ {\prime}) = R(\bar{P})$ (i.e. $\bar{P}$ with degenerate points removed). Then the new pair $(f, P ^ {\prime})$ is $F$-generatable.
\end{enumerate}
\end{definition}
To explain Operation 3, suppose $P = \{A, B, C, D, E, F, G\}$ as shown in \cref{fig:mps}, and suppose $\delta$ is equal to the difference in y-coordinates between $C$ and $F$. For $(x, y) = (C, F)$, since both $D$ and $E$ are below the line segment connecting $C$ and $F$, we can remove them and get $\bar{P} = \{A, B, C, F, G\}$. Note that since $F$ is on the segment connecting $C$ and $G$, it is a degenerate point of $\bar{P}$ and is removed from $\bar{P}$ to obtain $P ^ {\prime} = \{A, B, C, G\}$. We can not perform Operation 3 for $(x, y) = (A, F)$ despite all points between them being below the segment since $F_y - A_y > \delta$. 

It is easy to see that if $(f, P)$ is $F$-generatable for some $F$, then $P$ contains no degnerate points.

Solving \cref{prob:knapmain} is equivalent to combining $n$ profit functions $\bar{F} = \{\bar{f}[1], \bar{f}[2], \cdots, \bar{f}[n]\}$ where $\bar{f}[i]$ has a domain of $[0, w_i]$ and $\bar{f}[i](x) := 0$ for $0 \le x < w_i$ and $\bar{f}[i](x) := p_i$ for $x = w_i$. We can see that $f_I = \bar{f}[1] \oplus \bar{f}[2] \oplus \cdots \oplus \bar{f}[m]$. In our algorithm, since we can not afford to combine the actual profit functions themselves, we will instead combine monotone point sets and reduce the combined monotone point sets using Operation 3 whenever necessary. \Cref{def:fgen} provides us a tool to analyze our algorithm. Every time the algorithm computes a monotone point set $P$, the set is intended to approximate some unknown profit function $f$. Our algorithm will make sure that at any time, $(f, P)$ is $F$-generatable for some set $F$ that can be used to approximate $f_I$ (e.g. $\bar{F}$). Our goal is to show:
\begin{itemize}
    \item Correctness: an $F$-generatable $(f, P)$ should be such that $P$ can be used to approximate $f$.
    \item Efficiency: if $(f_a, P_a)$ and $(f_b, P_b)$ are both $F$-generatable, then there is an efficient way to obtain a $P$ such that $(f_a \oplus f_b, P)$ is $F$-generatable. This allows us to efficient approximate $f_a \oplus f_b$ if both functions are already approximated by $P_a$ and $P_b$.
\end{itemize}

For two functions $\tilde{f}$ and $f$ and a constant $c$, we say that $\tilde{f}$ $c$-subtractively approximates $f$ if $f$ $c$-additively approximates $\tilde{f}$. If $\tilde{f}$ $c$-subtractively approximates $f$, then the function $\hat{f}$ with identical domain as $\tilde{f}$ defined as $\hat{f}(x) := \tilde{f}(x) - \delta$ $c$-additively approximates $f$. Thus subtractive approximation is as powerful as additive approximation. The following shows that individual rounding errors do not accumulate:
\begin{restatable*}{lemma}{lemmageomain}\label{lemma:geomain}
    Given a set $F$ of profit functions. For any set of $m$ profit functions $\{f_1, f_2, \cdots, f_m\}$ and any set of $m$ monotone point sets $\{P_1, P_2, \cdots P_m\}$ such that $(f_1, P_1), (f_2, P_2), \cdots, (f_m, P_m)$ are all $F$-generatable, $(g[P_1] \oplus g[P_2] \oplus \cdots \oplus g[P_m])$ $\delta$-subtractively approximates $f_1 \oplus f_2 \oplus \cdots \oplus f_m$.
\end{restatable*}

For efficiency, we will first define a measure of complexity on monotone point sets. Given two monotone point sets $U = \{u_1, u_2, \cdots, u_{\abs{U}}\}$ and $V = \{v_1, v_2, \cdots, v_{\abs{V}}\}$, if $u_{\abs{U}} = v_1$, define the concatenation $U \circ V := \{u_1, u_2, \cdots, u_{\abs{U}} = v_1, v_2, \cdots v_{\abs{V}}\}$. A monotone point set $P = \{p_1, p_2, \cdots, p_k\}$ is an \emph{upper convex hull} if for all $1 < i < k$, the point $p_i$ is strictly above the line segment connecting $p_{i - 1}$ and $p_{i + 1}$. $P$ is a \emph{lower convex hull} if for all $1 < i < k$, the point $p_i$ is strictly below the line segment connecting $p_{i - 1}$ and $p_{i + 1}$. A \emph{decomposition} of monotone point set $P$ into $b$ blocks is a way to write $P = B[1] \circ B[2] \circ \cdots \circ B[b]$ such that each $B[i]$ is either an upper convex hull or a lower convex hull. For example, the monotone point set $P = \{A, B, C, D, E, F, G\}$ in \cref{fig:mps} can be decomposed into 3 blocks $P = \{A, B, C\} \circ \{C, D, E, F\} \circ \{F, G\}$, where $\{A, B, C\}$, $\{C, D, E, F\}$ are lower convex hulls, and $\{F, G\}$ can be considered either an upper convex hull or a lower convex hull.

Given a monotone point set $P = \{p_1, p_2, \cdots, p_k\}$ with no degenerate point, let $t_i$ be $1$ if $p_i$ is an apex and $0$ if $p_i$ is a base. Let the \emph{convex number} of $P$ be the number of maximal consecutive subsequences of $1$'s and $0$'s of the sequence $T = \{t_i\}$. For example, for the monotone point set $P = \{A, B, C, D, E, F, G\}$ in \cref{fig:mps}, we have $T = \{1, 0, 1, 0, 0, 1, 1\}$, which consists of $5$ maximal consecutive subsequences of $1$'s and $0$'s: $\{t_1\}, \{t_2\}, \{t_3\}, \{t_4, t_5\}$ and $\{t_6, t_7\}$ and thus the convex number of $P$ is $5$. For the relation between convex number and decomposition, we have the following observation:
\begin{obs} \label{obs:decompose}
Let $\tau$ be the convex number of $P$. We have a decomposition into $\tau$ blocks: for $1 \le i \le k$, there is a block starting from $p_i$ for if and only if $i = 1$ or $t_i \ne t_{i - 1}$.
\end{obs}
For example, for the monotone point set $P = \{A, B, C, D, E, F, G\}$ in \cref{fig:mps}, this gives the valid decomposition $P = \{A, B\} \circ \{B, C\} \circ \{C, D\} \circ \{D, E, F\} \circ \{F, G\}$.

We first show that there is a way to reduce the convex number of a monotone point set while still being able to approximate the associated profit function:
\begin{restatable*}{lemma}{lemmareduce}\label{lemma:reduce}
    Fix a set $F$ of profit functions. Suppose $(f, P)$ is $F$-generatable and $P$ contains $l$ points. We can compute another point set $P ^ {\prime}$ with a convex number of $O(({\max}_y{P} - {\min}_y{P}) / \delta)$ in $O(l)$ time such that $(f, P ^ {\prime})$ is $F$-generatable.
\end{restatable*}

We say a monotone point set $\hat{P}$ $\sigma$-accurately approximates another monotone point set if the upper boundary of $R(\hat{P})$ and $R(P)$ match whenever the y-coordinate is a multiple of $\sigma$, namely if:
\begin{itemize}
    \item $\floor{{\min}_y{\hat{P}} / \sigma} = \floor{{\min}_y{P} / \sigma}$ and $\floor{{\max}_y{\hat{P}} / \sigma} = \floor{{\max}_y{P} / \sigma}$.
    \item For every $y \in [{\min}_y{P}, {\max}_y{P}]$ such that $\sigma \mid y$, the minimum x-coordinate such that $g[P](x) \ge y$ is the same as the minimum x-coordinate such that $g[\hat{P}](x) \ge y$.
\end{itemize}

We show that if two monotone point sets have low convex numbers, we can combine them efficiently:
\begin{restatable*}{lemma}{lemmacombine} \label{lemma:combine}
    Suppose we have two monotone point sets $P_a$ and $P_b$. Suppose the y-coordinates of points in both $P_a$ and $P_b$ are multiples of $\sigma$ in $[0, l \sigma)$, and suppose both $P_a$ and $P_b$ have a convex number of at most $m$. We can compute a set that $\sigma$-accurately approximates $P_a \oplus P_b$ in $\tilde O(lm)$ time.
\end{restatable*}
The problem we are solving here is different from the one in \cite{brand_et_al:LIPIcs.ISAAC.2023.16} where all blocks must be upper convex hulls\footnote{See the part regarding convex numbers. The definition in the paper is that blocks are lower convex hulls, but the paper deals with $(\min, +)$-convolution instead of $(\max, +)$-convolution.}.

\Cref{lemma:reduce} and \cref{lemma:combine} give us a good way to achieve efficiency: suppose $(f_a, P_a)$ and $(f_b, P_b)$ are both $F$-generatable for some $F$. We first use \cref{lemma:reduce} to obtain $(f_a, {P ^ {\prime}}_a)$ and $(f_b, {P ^ {\prime}}_b)$ that are both $F$-generatable where both ${P ^ {\prime}}_a$ and ${P ^ {\prime}}_b$ have low convex numbers, and then we will \cref{lemma:combine} to obtain ${P ^ {\prime}}_a \oplus {P ^ {\prime}}_b$ ($\sigma$-accurately with regards to some parameter $\sigma$). By definition, $(f_a \oplus f_b, {P ^ {\prime}}_a \oplus {P ^ {\prime}}_b)$ is $F$-generatable.

With these lemmas, we can prove \cref{lemma:knapmain}.
\knapsackmainrestate
\begin{proof} [Proof of \cref{lemma:knapmain}]
    First, we can reduce $\eps$ so that $1/\eps$ becomes an integer.

    We can assume $n = 2 ^ K$ since otherwise we can pad the set with items of very large weights. For $i = K, K - 1, \cdots, 0$ in order, we will compute a set of profit functions $F_i = \{f_i[1], f_i[2], \cdots, f_i[2 ^ i]\}$ such that each function has values being multiples of $\eps 2 ^ {K - i}$ in $[0, 2 ^ {K - i + 1}]$ and that for each $i$, $f_i[1] \oplus f_i[2] \oplus \cdots f_i[2 ^ i]$ $(2 \eps n(K - i + 1))$-subtractively approximates $f_I$ for $x \in [0, \sum_{j \in I}{w_j}]$. By extending the domain of $f_0[1]$ to all $x \ge 0$ trivially we can get a $(2 \eps n \log_2{n})$-subtractive approximation of $f_I$, which can be made $\eps n$ by lowering the value of $\eps$ by a logarithmic factor.

    For $i = K$, it suffices to let $f_i[j]$ be $\bar{f}[j]$ with values rounded up to multiples of $\eps$, which introduces a subtractive error of $n \eps$. The values of $\bar{f}[j]$'s are contained in $[0, 2]$ since $p_j \in [1, 2]$ for $j \in I$. Since $1 / \eps$ is an integer, we never round a value over $2$, and thus the values of $f_i[j]$'s are contained in $[0, 2]$.

    Suppose we correctly computed $F_{i + 1}$. For each $J \in [1, 2 ^ {i + 1}]$. Let $\hat{P}_{i + 1}[J] = P(f_{i + 1}[J])$. Note that since values of $f_{i + 1}[J]$ are multiples of $\eps 2 ^ {K - i}$ no more than $2 ^ {K - (i + 1) + 1}$, the y-coordindates in $\hat{P}_i[J]$ are also multiples of $\eps 2 ^ {K - i}$ no more than $2 ^ {K - (i + 1) + 1}$, and thus $\hat{P}_i[J]$ has no more than $O(2 ^ {K - (i + 1) + 1} / (2 ^ {K - i}\eps)) = O(1 / \eps)$ points. We use \cref{lemma:reduce} on each $\hat{P}_{i + 1}[J]$ to get ${P ^ {\prime}}_{i + 1}[J]$ with a convex number of $O(2 ^ {K - (i + 1) + 1} / \delta))$, which takes total time $$O(2 ^ {i + 1} (1 / \eps)) = O(\eps ^ {-2}).$$ 
    We then use \cref{lemma:combine} to $(\eps 2 ^ {K - (i + 1)})$-accurately approximate $P_i[j] = {P ^ {\prime}}_{i + 1}[2j - 1] \oplus {P ^ {\prime}}_{i + 1}[2j]$ for each $j \in [1, 2 ^ i]$ in time 
    $$\tilde O(2 ^ i \times (1 / \eps) \times (2 ^ {K - (i + 1) + 1} / \delta)) = \tilde O(2 ^ k\eps ^ {-2} / n) = \tilde O(\eps ^ {-2}).$$
    We set ${\hat{F}}_i = \{g[P_i[1]], g[P_i[2]], \cdots, g[P_i[2 ^ i]]\}$. 
    For each $j$, we let $\tilde{f}_i[j] = f_{i + 1}[2j - 1] \oplus f_{i + 1}[2j]$. Then each $(\tilde{f}_i[j], P_i[j])$ is $F_{i + 1}$-generatable. By \cref{lemma:geomain} we can see that $g[P_i[1]] \oplus g[P_i[2]] \oplus  \cdots \oplus g[P_i[2 ^ i]]$ $\eps n$-subtractively approximates $f_i[1] \oplus f_i[2] \oplus \cdots \oplus f_i[2 ^ i]$. Thus ${\hat{F}}_i$ only introduces an extra additive error of $\eps n$ compared to $F_{i + 1}$. To get the set $F_i = \{f_i[1], f_i[2], \cdots, f_i[2 ^ i]\}$, we round values of each $g[P_i[j]]$ up to multiples of $\eps 2 ^ {K - i}$ to obtain $f_i[j]$, which is possible from a $(\eps 2 ^ {K - (i + 1)})$-accurate approximation of $P_i[j]$ since $\eps 2 ^ {K - i}$ is a multiple of $(\eps 2 ^ {K - (i + 1)})$. Such rounding introduces an additive error of at most $\eps 2 ^ {K - i} \times 2 ^ i = \eps n$. The total additive error for $F_i$ is at most: $$2 \eps n(K - (i + 1) + 1) + \eps n + \eps n = 2\eps n(K - i + 1).$$ 
    It is easy to see that $${\max}_y{P_i[j]} = {\max}_y{P_{i + 1}[2j - 1]} + {\max}_y{P_{i + 1}[2j]} \le 2 ^ {K - (i + 1) + 1} \times 2 = 2 ^ {K - i + 1},$$ and $${\min}_x{P_i[j]} = {\min}_x{P_{i + 1}[2j - 1]} + {\min}_x{P_{i + 1}[2j]} \ge 0.$$ Thus the values of the function $g[P_i[j]]$ are contained in $[0, 2 ^ {K - i + 1}]$. Since $1 / \eps$ is an integer, we never round a value over $2 ^ {K - i + 1}$. Thus the values of the function $f_i[j]$ are also contained in $[0, 2 ^ {K - i + 1}]$. Since there are $K = \log_2{n} = \tilde O(1)$ layers, the total running time is $\tilde O(\eps ^ {-2})$.
\end{proof}

\subsubsection{Proof of \cref{lemma:geomain}, \cref{lemma:reduce} and \cref{lemma:combine}}
For correctness, we need to prove \cref{lemma:geomain}:
\lemmageomain*
To derive \cref{lemma:geomain}, we need to prove some ingredients. 
\begin{lemma} \label{lemma:apex}
    For any two monotone point sets $P$ and $Q$ with no degenerate points, let $V = P \oplus Q$. Then for every $v$ on the upper boundary of $R(V)$, $v = p + q$ such that either:
    \begin{itemize}
        \item $p$ is on the upper boundary of $P$ and $q$ is an apex of $Q$, or
        \item $q$ is on the upper boundary of $Q$ and $p$ is an apex of $P$.
    \end{itemize}
\end{lemma}
\begin{proof}
    We know that for every $v$ on the upper boundary of $R(V)$, $v = p + q$ such that either:
    \begin{itemize}
        \item $p$ is on the upper boundary of $P$ and $q \in Q$, or
        \item $q$ is on the upper boundary of $Q$ and $p \in P$.
    \end{itemize}
    Without loss of generality suppose that $p$ is on the upper boundary of $P$ and $q \in Q$. 
    
    If $p$ is an apex of $P$, then since $q$ is on the upper boundary of $Q$, we are already done. Otherwise, it suffices to show that $q$ is an apex of $Q$.
    
    Since $p$ is not an apex of $P$, we have an infinitesimal segment $e \subset R(P)$ such that $p$ is strictly contained in $e$. Thus $p + q$ is strictly contained in $e + q$, which is contained in $R(P) + q$. Suppose that $e$ connects $p_l$ and $p_r$ with $p_l < p < p_r$. 
    
    \begin{figure}[H]
    \centering 
        \begin{tikzpicture}[scale=0.8]

    % Define points
    \coordinate (A) at (0,0);
    \coordinate (B) at (8,8);
    \coordinate (C) at (0,3);
    \coordinate (D) at (5,8);
    \coordinate (E) at (4,4);
    \coordinate (F) at (-2,2);
    \coordinate (G) at (-2,4.666667);
    \coordinate (H) at (2,6);
    % Draw shaded area
    \fill[gray!30] (A) -- (8,0) -- (B) -- (D) -- (E) -- (C) -- cycle;

    % Draw curve
    \draw[black] (C) -- (E) -- (D);
    \draw[<-, dashed, black] (F) -- (E);
    \draw[<-, dashed, black] (G) -- (B); 
    \draw[thick,black] (A) -- (B);
    \node[right=2pt] at (2, 2) {\large $\mathbf{e + q}$};
    \node[above=2pt] at (6, 0.5) {\LARGE $R(P) + q$};

    % Label points
    \foreach \point/\position/\name in {A/below/{$p_l + q$}, E/right/{$p + q$}, B/right/{$p_r + q$}, H/above/{$p ^ {\prime}$}} {
        \fill[black] (\point) circle (5pt);
        \node[\position=5pt] at (\point) {\name};
    }
    \foreach \point/\position/\name in {F/left/{To $p + q_l$}, G/left/{To $p_r + q_l$}} {
        \node[\position=5pt] at (\point) {\large \name};
    }

\end{tikzpicture}
    \caption{Proof of \cref{lemma:apex}} \label{fig:apex}
    \end{figure}  

    \Cref{fig:apex} shows the immediately neighborhood of $p + q$ after zooming into the infinitesimal segment $e + q$, the shaded region is $R(P) + q$. Suppose $q$ is not an apex of $Q$, then $q$ is not an endpoint of $Q$. Suppose $q_l$ and $q_r$ are the point before and after $q$ in $Q$. The points $p + q_l$ and $p + q_r$ can not both be below $e + q$ since otherwise it would make $q$ an apex of $Q$. Without loss of generality suppose $p + q_l$ is above the segment $e + q$, as shown in the figure. Then a non-zero fraction of the segment connecting $p_r + q_l$ and $p_r + q$, which is entirely in $R(V)$, is above $p + q$. Thus as shown in the figure, we can find a point $p ^ {\prime}$ with ${p ^ {\prime}}_x < p_x$ and ${p ^ {\prime}}_y > p_y$ on the segment. Since $p ^ {\prime} \in R(V)$, $v = p + q$ must not be on the upper boundary of $R(V)$ and thus must not be in++ $V$, which is a contradiction.    
\end{proof}

\begin{lemma} \label{lemma:apex2}
    For any two monotone point set $P$ and $Q$ with no degenerate points, if we let $V = P \oplus Q$, then every apex $v$ of $V$ must be the sum of an apex $p$ of $P$ and an apex $q$ of $Q$.
\end{lemma}
\begin{proof}
    Since $v$ is an apex of $V$, it is easy to see that there must not be a non-degenerate segment (i.e. segment with a length greater than $0$) $e \subset V$ such that $v \in e$.

    Suppose $v = p + q$. Suppose $p$ is not an apex of $P$. Then there exists an infinitesimal segment $e \subset R(P)$ such that $p$ is strictly contained in $e$. We have a non-degenerate segment $e + q \subset V$ that contains $v = p + q$, which is a contradiction. Therefore $p$ must be an apex of $P$. Symmetrically, $q$ must be an apex of $Q$.  
\end{proof}
\begin{lemma} \label{lemma:invariants}
    For any set $F$ of profit functions, and an $F$-generatable $(f, P)$, we have the following:
    \begin{enumerate}
        \item For every $x$ in the domain of $f$, $g[P](x) \ge f(x)$. Equivalently, $R(P(f)) \subset R(P)$.
        \item For any apex $p \in P$, $f(p_x) = p_y$.
        \item For any point $p$ on the upper boundary of $R(P)$, there exist points $a$ and $b$ $(0 \le a_y < b_y \le a_y + \delta)$ such that $f(a_x) \ge a_y, f(b_y) \ge b_y$, and $p$ is on the line segment connecting $a$ and $b$.
    \end{enumerate}
    
\end{lemma}
\begin{proof}
    It suffices to argue that each of the three operations in \cref{def:fgen} does not introduce new $F$-generatable pairs that violate each of the invariants.

    For operation (1), since $P = P(f)$, Invariants (1), (2) are obviously satisfied. For Invariant (3), since we must have $f(p_x) \ge p_y$, we can set $a = b = p$.

    For operation (2), fix $(f_a, P_a)$ and $(f_b, P_b)$. Since $R(P(f_a)) \subset R(P_a)$ and $R(P(f_b)) \subset R(P_b)$, it is easy to see that $R(P(f_a \oplus f_b)) \subset R(P_a) + R(P_b) = R(P_a \oplus P_b)$, and therefore invariant (1) must hold. For Invariant (2), suppose a point $p = (p_x, p_y)$ is an apex of $P_a \oplus P_b$. From \cref{lemma:apex2}, $p$ can be expressed as the sum of two apexes $u$ on $P_a$ and $v$ on $P_b$. Since $f_a(u_x) = u_y$ and $f_b(v_x) = v_y$, we must have $p_y = u_y + v_y = f_a(u_x) + f_b(u_y) \le (f_a \oplus f_b)(p_x)$. Since $p_y = g[P](p_x) \ge (f_a \oplus f_b)(p_x)$ (from Invariant (1)), we must have $(f_a \oplus f_b)(p_x) = p_y$. For Invariant (3), let $p$ be a point on the boundary of $R(P)$, then from \cref{lemma:apex}, $p = u + v$ such that either $u$ is on the upper boundary of $R(P_a)$ and $v$ is an apex of $P_b$, or $v$ in on the upper boundary of $R(P_b)$ and $u$ is an apex of $P_a$. Without loss of generality suppose that  $u$ is on the upper boundary of $R(P_a)$ and $v$ is an apex of $P_b$. Then $f_b(v_x) = v_y$ (Invariant (2)). Note that Invariant (3) is satisfied for $(f_a, P_a)$. Let $a_u$ and $b_u$ be the points in Invariant (3) for $u$ on the upper boundary of $R(P_a)$. Then it is not hard to verify that setting $a := a_u + v$ and $b := b_u + v$ satisfies Invariant (3) for $p$ on the upper boundary of $R(P)$.

    For operation (3), since we only remove non-apexes in $P$ to obtain $P ^ {\prime}$, $R(P) \subset R(P ^ {\prime}) = R(\hat{P})$. Therefore $P(f) \subset R(P) \subset R(\hat{P})$, and thus Invariant (1) must hold. It is easy to see that no new apexes can be created, and therefore Invariant (2) does not get broken. Consider Invariant (3) for $\hat{P}$. If $p$ is on the upper boundary of $P$, then Invariant (3) holds since it holds for $(f, P)$. Otherwise, it is easy to see that $p$ must be on a segment connecting two apexes $a$ and $b$ with y-coordinates at most $\delta$ apart and it is easy to verify that $a$ and $b$ satisfy Invariant (3) for $p$.
\end{proof} 

Recall that $\tilde{f}$ $c$-subtractively approximates $f$ if $f$ $c$-additively approximates $\tilde{f}$. 
\begin{corollary} \label{corollary:approximate}
    For any set $F$ of profit functions, and an $F$-generatable $(f, P)$, $g[P]$ $\delta$-subtractively approximates $f(x)$.
\end{corollary}
\begin{proof}
    We show that for every $x$ in the domain of $f$, we have $g[P](x) - \delta \le f(x) \le g[P](x)$.

    $f(x) \le g[P](x)$ follows from Invariant (1) of \cref{lemma:invariants}. To see why $g[P](x) - \delta \le f(x)$. Apply Invariant (3) on $(x, g[P](x))$ and obtain $a$ and $b$. Since $a_y \le g[P](x) \le b_y$ and $f$ is non-decreasing, we have $f(x) \ge f(a_x) \ge a_y \ge b_y - \delta \ge g[P](x) - \delta$.
\end{proof}

Now we are ready to prove \cref{lemma:geomain}:
\lemmageomain
\begin{proof}
    Since with $m - 1$ repetitions of Operation 2 we can $F$-generate $(f_1 \oplus f_2 \oplus \cdots \oplus f_m, P_1 \oplus P_2 \oplus \cdots \oplus P_m)$, $g[P_1 \oplus P_2 \oplus \cdots \oplus P_m]$ $\delta$-subtractively approximates $f_1 \oplus f_2 \oplus \cdots \oplus f_m$. It is easy to see that $g[P_1] \oplus g[P_2] \oplus \cdots \oplus g[P_m] = g[P_1 \oplus P_2 \oplus \cdots \oplus P_m]$. Therefore $(g[P_1] \oplus g[P_2] \oplus \cdots \oplus g[P_m])$ $\delta$-subtractively approximates $f_1 \oplus f_2 \oplus \cdots \oplus f_m$.
\end{proof}

Now we move on to efficiency and prove \cref{lemma:reduce}:
\lemmareduce
\begin{proof}
    %Consider doing the following modification of the famous Graham's Scan \cite{Graham1972AnEA} on $P$ to get $P ^ {\prime}$. 
    We can obtain $P ^ {\prime}$ in the following way. Initially set $P ^ {\prime} = P$. Scan the apexes in order and store them in a stack. Every time we encounter a new apex $x$. If the stack is not empty, the top of the stack is $y$, and $\abs{x_y - y_y} \le \delta$, we can simulate an Operation 3 on $P ^ {\prime}$ with $x$ and $y$, which involves:
    \begin{itemize}
        \item Removing the vertices between $x$ and $y$ in $P ^ {\prime}$, and
        \item If $y$ becomes degenerate, removing $y$ from $P ^ {\prime}$.
    \end{itemize}
    After this, if $y$ is no longer an apex of $P ^ {\prime}$ or gets removed from $P ^ {\prime}$, we remove $y$ from the stack and repeat the procedure above. Otherwise, we push $x$ to the end of the stack.

    After processing each apex $x$, let ${P ^ {\prime}}_{1 \cdots k}$ be the prefix of $P ^ {\prime}$ containing the remaining points in $P ^ {\prime}$ up to $x$. It is easy to verify that the following two invariants are true:
    \begin{itemize}
        \item No points of ${P ^ {\prime}}_{1 \cdots k}$ are degenerate.
        \item The stack contains exactly all apexes of ${P ^ {\prime}}_{1 \cdots k}$.
        \item For every maximal consecutive subsequence of bases in ${P ^ {\prime}}_{1 \cdots k}$ enclosed by two apexes $a < b$, $b_y - a_y > \delta$.
    \end{itemize}
    
    The procedure can easily be done in $O(l)$ time, and the resulting $P ^ {\prime}$ is clearly $F$-generatable. Since every maximal consecutive subsequence of bases in ${P ^ {\prime}}$ is enclosed by two apexes $\delta$ with a difference larger than $\delta$ in y-coordinates, it is easy to see that $P ^ {\prime}$ has a convex number of $O(({\max}_y{P} - {\min}_y{P}) / \delta)$.
\end{proof}

Finally, we prove \cref{lemma:combine}: 
\lemmacombine
\begin{proof}
    We can suppose $\sigma = 1$ since otherwise we can re-scale all the y-coordinates. We obtain decompositions $P_a = B_a[1] \circ B_a[2] \circ \cdots \cdots B_a[m_a]$ with $m_a \le m$ and $P_b = B_b[1] \circ B_b[2] \circ \cdots \circ B_b[m_b]$ with $m_b \le m$ from \cref{obs:decompose}.

    Suppose that for each two blocks $B_a[i] \in P_a$ and $B_b[j] \in P_b$, we can compute $A[i, j] = B_a[i] \oplus B_b[j]$. Then by taking the union of $R(A[i, j])$'s, we get $R(P_a \oplus P_b)$ which can give us $P_a \oplus P_b$. The total number of y-coordinates involved in $R(A[i, j])$ does not exceed $({\max}_y{B_a[i]} - {\min}_y{B_a[i]}) + ({\max}_y{B_b[j]} - {\min}_y{B_b[j]}) + 2$. Since we only need a $1$-accurate approximation, to take the union of $R(A[i, j])$'s, we only need to find the minimum x-coordinate for every y-coordinate. Therefore every time we take the union with a new $R(A[i, j])$, we only need to go through all the involved y-coordinates and update the minimum x-coordinate. For all $R(A[i, j])$'s in total this takes time $$O(\sum_{i, j}{({\max}_y{B_a[i]} - {\min}_y{B_a[i]}) + ({\max}_y{B_b[j]} - {\min}_y{B_b[j]}) + 2}) = O(lm).$$
    
    We claim that the following claim is true: For two blocks $B_a[i]$ and $B_b[j]$ with the total span of y-coordinates being $\mathcal{Y}$, $B_a[i] \oplus B_b[j]$ can be computed in $\tilde O(\mathcal{Y})$ time. 
    
    If the claim is true, then the total running time for computing every $A[i, j]$ is $$\tilde O(\sum_{i, j}{({\max}_y{B_a[i]} - {\min}_y{B_a[i]}) + ({\max}_y{B_b[j]} - {\min}_y{B_b[j]}) + 2}) = \tilde O(lm),$$ which finishes the proof of the lemma. Now we prove that for two blocks $B_a[i]$ and $B_b[j]$ with the total span of y-coordinates being $\mathcal{Y}$, we can indeed compute $B_a[i] \oplus B_b[j]$ in $\tilde O(\mathcal{Y})$ time. 
    
    For the case where both $B_a[i]$ and $B_b[j]$ are upper convex hulls, it is well known that the Minkowski sum of two convex polygons with $v$ vertices in total can be computed in $\tilde O(v)$ time \cite{compgeom:2000}. Since the number of vertices is bounded by the number of different y-coordinates, such running time is sufficient.

    It suffices to deal with computations where one of the blocks is a lower convex hull. Without loss of generality, suppose $B_a[i]$ is a lower convex hull. Let $X = {\max}_y{B_a[i]} - {\min}_y{B_a[i]} + 1$, and $Y = {\max}_y{B_b[j]} - {\min}_y{B_b[j]} + 1$. We will now introduce a procedure that computes $B_a[i] \oplus B_b[j]$ in $\tilde O(X + Y)$ time.

    Without loss of generality, suppose the y-coordinates of $B_a[i]$ span $[0, X - 1]$ and those of $B_b[j]$ span $[0, Y - 1]$. For $y \in [0, X - 1]$, define $\phi(y)$ as the minimum $x$ with $g[B_a[i]](x) \ge y$. For $y \in [0, Y - 1]$, define $\gamma(y)$ as the minimum $x$ with $g[B_b[j]](x) \ge y$. It suffices to compute the convolution $$\eta(y) = \min_{y_f \in [0, X - 1], y_g \in [0, Y - 1], y_f + y_g = y}\phi(y_f)+\gamma(y_g),$$ where we know that $\phi$ is convex (i.e. the difference $\phi(y + 1) - \phi(y)$ decreases as $y$ increases), in time linear in $X$ and $Y$.

    For the computation of $\eta(y)$ for $y \in [0, X - 1]$, we use a ``reversed'' version of a well-known technique in competitive programming\footnote{See e.g. \url{https://codeforces.com/blog/entry/8219?\#comment-139241}. We believe that the reversed version has also been featured in some competitive programming task in the past, but unfortunately, we could not find a reference.}. We treat values of $\gamma(y)$ where $y > Y - 1$ to be $\infty$. We compute $\eta(y)$ in increasing order of $y$. We have different ``candidates'' $c$ which we can use to update $\eta(y)$ with $\phi(y - c) + \gamma(c)$. For two candidates $c_a$ and $c_b$ such that $c_a < c_b$, if $\phi(y - c_a) + \gamma(c_b) < \phi(y - c_b) + \gamma(c_b)$, then since $\phi(y - c_a)$ increases more slowly than $\phi(y - c_b)$ as $y$ increases, $c_b$ will never be a better candidate than $c_a$\footnote{In the more common, ``non-reversed'' version, the observation is true for $c_a > c_b$. Thus a double-ended queue is used, but here we will use a stack instead. The ``non-reversed'' version can also be directly solved using the SMAWK algorithm.}. For two candidates $c_a < c_b$, we can compute the minimum $\hat{y}$ at which $\phi(\hat{y} - c_a) + \gamma(c_a) < \phi(\hat{y} - c_b) + \gamma(c_b)$ (or lack thereof) by binary search, and such $\hat{y}$ is the ``time'' when $c_a$ becomes better than $c_b$. We maintain a stack that stores the candidates $c_1 < c_2 < c_3 < \cdots < c_k$ from bottom to top. We also maintain the timestamps $t_1 < t_2 < \cdots < t_{k - 1}$ where $t_i$ is the time at which $c_i$ becomes better than $c_{i + 1}$, namely the smallest $t_i$ such that $\phi(t_i - c_i) + \gamma(c_i) < \phi(t_i - c_{i + 1}) + \gamma(c_{i + 1})$. After computing $\eta(y)$ for $y = c$, $c$ becomes available as a new candidate for future $y$'s. we compute the time $\hat{t}_k$ at which the top of the stack $c_k$ becomes better than $c$. If $\hat{t}_k \ge t_{k - 1}$, then $c_k$ can never become the best candidate and we can remove it from the stack. We repeat until at some $k ^ {\prime}$ we have $t_{k ^ {\prime} - 1} < \hat{t}_{k ^ {\prime}}$ and then add $c$ as $c_{k ^ {\prime} + 1}$ into the stack and set $t_{k ^ {\prime}} = \hat{t}_{k ^ {\prime}}$. As each candidate gets removed at most once, the total running time is $\tilde O(X)$.

    This algorithm will not work for $y \ge X$ as a candidate $c$ becomes unavailable when $y - c > X - 1$. To deal with this issue, we split integers in $[0, X + Y - 1)$ into consecutive blocks of size $X - 1$, where block $k$ contains integers in $\{[k(X - 1), \min(X + Y - 1, (k + 1)(X - 1)))\}$. For each $y$, a candidate $c$ within the same block will always be available. Thus if we call the procedure above on each block we can cover all pairs of $(y, c)$ in the same block. It remains to deal with the case where $y$ and $c$ are in adjacent blocks. To do this, for every two adjacent blocks $[(k - 1)(X - 1), k(X - 1)), [k(X - 1), (k + 1)(X - 1))$, in decreasing order of $i = X - 2, X - 3, \cdots 0$, we compute the contribution of candidates in $[(k - 1)(X - 1) + i, k(X - 1))$ to $y = k(X - 1) + i$ (if $y < X + Y - 1$). To do this, we first add $c = (k - 1)(X - 1) + i$ to the set of available candidates and then find the best candidate for $y$ from the set. It is not hard to maintain and find the best candidate for all $y$'s in near-linear time using the same technique described before. The total running time is $\tilde O(X + Y)$.
\end{proof}

\section{Open Problems}
Although this paper closes the gap between the best algorithm and the conditional lower bound for the Knapsack problem, there are still some open problems in this field.
\begin{itemize}
    \item Can we make our algorithm simple? Intuitively, a ``simple'' algorithm should not use any results from additive combinatorics, or a technique analogous to Chan's number theoretical lemma. Despite the fact that our algorithm for \cref{lemma:knapmain} is simple, the reduction to the special case in the proof of \cref{lemma:greedy} requires \cref{lemma:exchange}, which uses results from additive combinatorics, as well as \cref{lemma:chan}, which uses the number theoretical lemma.
    \item Can we shave an $\Omega(2 ^ {\sqrt{\log ^ {-1} (1 / \eps)}})$ factor from our algorithm? The previous algorithm from \cite{ourprevious} has a running time of $\tilde O\left(n + (\frac{1}{\eps})^{11/5}/2^{\Omega(\sqrt{\log(1/\eps)})}\right)$, which has such a factor shaved. Moreover, can we establish an equivalence between $(1 - \eps)$-approximation of Knapsack and $(\min, +)$-convolution? Alternatively, can we prove a stronger conditional lower bound that rules out such an improvement?
%    \item Can we close the gap for \parti as well? Currently the best algorithm runs in $\tilde O(n+(\frac{1}{\eps})^{5/4})$ time \cite{ourprevious} but the condition lower bound is $\poly(n)/\eps^{1-\delta}$ \cite{abboud2019seth}, and there is still a substantial gap between $5 / 4$ and $1$.
\end{itemize}

	\bibliographystyle{alphaurl} 
	\bibliography{main}

 \appendix
\section{Proof of \cref{lemma:exchange}}
The following structural lemma follows from Theorem 4.1 and Theorem 4.2 of Bringmann and Wellnitz \cite{bw21}.
\begin{restatable}{lemma}{densityrestate}
    Let $n$ distinct positive integers $X=\{x_1,\dots,x_n\} \subseteq [\ell,2\ell]$ be given, where $\ell = o(n^2 /\log n)$.
            
    Then, for a universal constant $c\ge 1$, for every $c\ell^2/n \le t \le \Sigma(X)/2$, there exists $t ^ {\prime} \in \caS(X)$ such that $0\le t^\prime - t  \le 8\ell / n$.  
    \label{lemma:density}
\end{restatable}

Now we are ready to prove \cref{lemma:exchange}
\greedyexchange*
\begin{proof}
          If $D(i) <\Delta$, then by the definition of $i$ we have $i=n$, and we can simply let $\tilde S=S$, since $\tilde p=0$ always holds. So in the following we assume $D(i)=\Delta$.

                We define $\tilde S\subseteq [n]$ as the maximizer of 
                \[ \sum_{s\in \tilde S \cap [i]}p_s + \sum_{s\in \tilde S \cap ([n] \setminus [i])}(1-\eps)p_s\]
                among all $\tilde S$ satisfying $\sum_{s\in \tilde S} w_s \le \sum_{s\in S} w_s$ and $\sum_{s\in \tilde S} p_s \le \sum_{s\in S} p_s$.
                 We claim that $\tilde S$ satisfies the properties \eqref{eqn:req1}, \eqref{eqn:req2}, \eqref{eqn:req3}. Observe that \eqref{eqn:req2}, \eqref{eqn:req3} immediately follow from the definition of $\tilde S$. The main part is to verify \eqref{eqn:req1}.

Suppose for contradiction that \eqref{eqn:req1} does not hold. Then, we can find a subset $K\subseteq \tilde S \cap ([n]\setminus [i])$ with total profit $p^{*} = \sum_{k\in K} p_k \in (B,B+2]$, which can  be obtained by removing items from $\tilde S\cap ([n]\setminus [i])$ (recall that each item has profit in $[1,2)$).

        Define item set $I ^ {\prime} := [i]\setminus \tilde S$.
         Since $|\tilde S| < \sum_{s\in \tilde S}p_s/\min_{s\in \tilde S} p_s \le \sum_{s\in \tilde S}p_s \le \sum_{s\in  S}p_s \le 2m$, 
         by the definition of $D(i)$, we know that $\{p_i: i\in I ^ {\prime}\}$ contains at least $D(i)=\Delta$ distinct elements.

         Suppose $c$ is the universal constant in \cref{lemma:density}. We apply \cref{lemma:density} on the set of integers $X=\{p_i/\eps : i\in I^\prime \} \subseteq [1/\eps,2/\eps)$ which contains at least $\Delta$ distinct integers, where the premise $1/\eps = o(\Delta^2/\log \Delta)$ in \cref{lemma:density} is satisfied due to the fact that $\Delta = \omega(\eps ^ {-1 / 2} \log ^ {1 / 2}{(1 / \eps)})$. \cref{lemma:density} states that for every $t\in [c\eps^{-2}/\Delta, 0.5\Delta/\eps]$, there exists $t'\in \caS(X)$ such that $0\le t'-t\le 8\eps^{-1}/\Delta$. Here we set \[t:=\frac{(1-\eps) p^{*}}{\eps} + \frac{\eps^{-1}}{\Delta},\] which satisfies $t>p^{*}(1-\eps)/\eps > (1-\eps)B/\eps  = (1-\eps)(9c\eps^{-1}/\Delta)/\eps > c\eps^{-2}/\Delta$, and $t < (B+2)/\eps + \eps^{-1}/\Delta = 9c\eps^{-2}/\Delta + 2/\eps +\eps^{-1}/\Delta\le o(\eps^ {-3/2}\log ^ {-1/2}(1/\eps)) \le 0.5\Delta/\eps$. Then the conclusion of \cref{lemma:density} says that there is a subset  $R\subseteq I'$ of items with total profit $\tilde p:= \eps \cdot t'$, satisfying 
         \begin{equation}
            1/\Delta\le \tilde p - p^{*}(1-\eps)\le 9/\Delta. \label{eqn:my}
         \end{equation}
         Note that \eqref{eqn:my} implies
         \begin{align*}
             p^{*} - \tilde p &\ge \eps\cdot p^{*} - 9/\Delta\\
             & > \eps \cdot B -9/\Delta\\
             &  = \eps \cdot 9c \eps^{-1}/\Delta - 9/\Delta\\
             & \ge 0.
         \end{align*}

         Recall that
          $R\subseteq I' = [i]\setminus \tilde S$
           and $K\subseteq \tilde S \cap ([n]\setminus [i])$,
            which must both be non-empty.
          Since the efficiency of items are sorted in non-increasing order, we have $\min_{r\in R}p_r/w_r \ge \max_{k\in K} p_k/w_k$. Now we define the set of items \[\tilde S':= (\tilde S \setminus K) \cup R. \]
         Then, we have
         \begin{align*}
            \sum_{s\in  \tilde S}p_s - \sum_{s\in \tilde S'}p_s &= \sum_{k\in  K}p_k - \sum_{r\in R}p_r \\
            & = p^* -\tilde p\\
            & \ge 0,
         \end{align*}
         and
         \begin{align*}
            \sum_{s\in  \tilde S}w_s - \sum_{s\in \tilde S'}w_s &= \sum_{k\in  K}w_k - \sum_{r\in R}w_r \\
            & \ge \frac{\sum_{k\in K}p_k}{\max_{k\in K}(p_k/w_k)} - \frac{\sum_{r\in R}p_r}{\min_{r\in R}(p_r/w_r)}\\
            & \ge  \frac{1}{\min_{r\in R}(p_r/w_r)}\cdot \left ( \sum_{k\in K}p_k - \sum_{r\in R}p_r\right )\\
 &  = \frac{1}{\min_{r\in R}(p_r/w_r)}\cdot \left ( p^{*} - \tilde p\right )\\
 & \ge 0.
         \end{align*}
         Hence, $\sum_{s\in \tilde S'}p_s \le \sum_{s\in \tilde S}p_s$ and $\sum_{s\in \tilde S'}w_s \le \sum_{s\in \tilde S}w_s$.
          On the other hand, by \eqref{eqn:my}, we know that
         \begin{align*}
 &                \left (\sum_{s\in \tilde S' \cap [i]}p_s + \sum_{s\in \tilde S' \cap ([n] \setminus [i])}(1-\eps)p_s\right ) - \left (\sum_{s\in \tilde S \cap [i]}p_s + \sum_{s\in \tilde S \cap ([n] \setminus [i])}(1-\eps)p_s\right )\\
 = \ &  \sum_{r\in R} p_r - \sum_{k\in K}(1-\eps) p_k\\
 = \ & \tilde p - (1-\eps)p^*\\
 \ge \ & 1/\Delta > 0,
 \end{align*}
 contradicting the definition of $\tilde S$ being a  maximizer.
 
 Hence, we have established that $\tilde S$ satisfies \eqref{eqn:req1}.
            \end{proof}

\section{Proof of \cref{lem:reduction-knap}}
Recall that we defined the following simpler problem.
\proknaprestate*

\redknap*

\begin{proof}
First, we can reduce $\eps$ so that $1/\eps$ becomes an integer.

 We will restrict the profit values into small intervals, as follows:  divide the items into $O(\log \frac{\max_j p_j}{\min_j p_j}) =$ $O(\log \eps^{-1})$ groups (see \cref{sec:prelimknap}), each containing items with $p_i \in [2^{j},2^{j+1}]$ for some $j$ (which can be rescaled to $[1,2]$). Finally, use the merging lemma \cref{lemma:dc} to merge the profit functions of all groups, in $\tilde O(n+\eps^{-2})$ overall time.  
 
 Now, having restricted the profit values into $[1,2)$, we can  round every profit value to a multiple of $\eps$, which incurs only $(1-O(\eps))$ approximation factor in total.
 
 Finally, the following greedy lemma takes care of the case with total profit above $\Omega(\eps^{-1})$. 
\begin{lemma}
\label{sortgreed}
 Suppose $p_i \in [1,2]$ for all $i\in I$. For $B=\Omega(\eps^{-1})$,  the profit function $f_I$ can be approximated with  additive error $O(\eps B)$ in $O(n\log n)$ time.
\end{lemma}
\begin{proof}
Simply sort the items in nonincreasing order of efficiency $p_i/w_i$, and define the profit function $\tilde f$ resulting from greedy, with function values $0,p_1,p_1+p_2,\dots,p_1 + \dots + p_n$ and $x$-breakpoints $0,w_1,w_1+w_2,\dots,w_1+\dots+w_n$.  It clearly approximates $f_I$ with an additive error of $\max_i p_i \le 2 \le O(\eps B)$ for $B=\Omega(\eps^{-1})$. 
\end{proof}
This greedy approach achieves $(1-O(\eps))$-approximation for large profit values. Hence, it is sufficient to approximate $f_I$ up to $2/\eps$.
\end{proof}

\end{document}